\pdfoutput=1

\documentclass[sigconf, nonacm]{acmart}

\usepackage[utf8]{inputenc}

\setcopyright{none}
\usepackage{graphicx}

\usepackage{preamble}

\newcommand{\B}{\ensuremath{\mathcal{B}}}
\newcommand{\UU}{\ensuremath{\mathcal{U}}}
\newcommand{\Q}{\ensuremath{\mathcal{Q}}}
\newcommand{\D}{\ensuremath{\mathcal{D}}}
\newcommand{\DQ}{\ensuremath{\mathcal{D}\mathcal{Q}}}
\newcommand{\T}[0]{\ensuremath{\mathcal{T}}}
\newcommand{\clean}[0]{\ensuremath{\textsc{Clean}}}

\newcommand{\chain}{\ensuremath{\mathsf{C}}}

\newcommand{\aux}{\ensuremath{\mathcal{L}}}

\newcommand{\client}{\ensuremath{\mathsf{c}}}

\newcommand{\PI}[0]{\ensuremath{\Pi}}

\newcommand{\Adv}[0]{\ensuremath{\mathcal A}}

\newcommand{\ie}[0]{\emph{i.e.}\xspace}
\newcommand{\eg}[0]{\emph{e.g.}\xspace}
\newcommand{\cf}[0]{\emph{cf.}\xspace}

\newcommand{\GST}[0]{\ensuremath{\mathsf{GST}}}

\newcommand{\tx}[0]{\ensuremath{\mathsf{tx}}}

\newcommand{\close}[0]{\ensuremath{\operatorname{cl}}}
\newcommand{\poly}[0]{\ensuremath{\operatorname{poly}}}
\newcommand{\Tconfirm}[0]{\ensuremath{T_{\mathrm{fin}}}}
\newcommand{\ckpt}{\ensuremath{\mathrm{ckpt}}}

\newcommand{\bprop}[1]{%
    \ifthenelse{\equal{#1}{}}{%
        \ensuremath{\Hat{b}}%
    }{%
        \ensuremath{\Hat{b}_{#1}}%
    }%
}

\newcommand{\ld}[1]{%
    \ifthenelse{\equal{#1}{}}{%
        \ensuremath{\mathrm{L}^{(c)}}%
    }{%
        \ensuremath{\mathrm{L}^{(#1)}}%
    }%
}

\newcommand{\fS}[0]{\ensuremath{f_{\mathrm{s}}}}

\newcommand{\fL}[0]{\ensuremath{f_{\mathrm{l}}}}

\definecolor{myParula01Blue}{RGB}{0,114,189}
\definecolor{myParula02Orange}{RGB}{217,83,25}
\definecolor{myParula03Yellow}{RGB}{237,177,32}
\definecolor{myParula04Purple}{RGB}{126,47,142}
\definecolor{myParula05Green}{RGB}{119,172,48}
\definecolor{myParula06LightBlue}{RGB}{77,190,238}
\definecolor{myParula07Red}{RGB}{162,20,47}

\begin{document}

\title{Interchain Timestamping for Mesh Security}

\author{Ertem Nusret Tas}
\affiliation{Stanford University\country{}}
\email{nusret@stanford.edu}

\author{Runchao Han}
\affiliation{BabylonChain\country{}}
\email{runchao.han@babylonchain.io}

\author{David Tse}
\affiliation{Stanford University\country{}} 
\email{dntse@stanford.edu}

\author{Fisher Yu}
\affiliation{BabylonChain\country{}}
\email{fisher.yu@babylonchain.io}

\author{Kamilla Nazirkhanova}
\affiliation{Stanford University\country{}} 
\email{nazirk@stanford.edu}

\begin{abstract}
Fourteen years after the invention of Bitcoin, there has been a proliferation of many permissionless blockchains. Each such chain provides a public ledger that can be written to and read from by anyone.  In this multi-chain world, a natural question arises: what is the optimal security an existing blockchain, a consumer chain,  can extract by only reading and writing to $k$ other existing blockchains, the provider chains? 
We design a protocol, called interchain timestamping, and show that it extracts the maximum economic security from the provider chains, as quantified by the slashable safety resilience.
We observe that interchain timestamps are already provided by light-client based bridges, so interchain timestamping can be readily implemented for Cosmos chains connected by the Inter-Blockchain Communication (IBC) protocol. We compare interchain timestamping with cross-staking, the original solution to mesh security, as well as with Trustboost, another recent security sharing protocol.
\end{abstract}

\maketitle

\section{Introduction}
\label{sec:introduction}
\subsection{Motivation}
Bitcoin, invented by Nakamoto in 2008 \cite{bitcoin}, is the first permissionless blockchain. It provides a public ledger, which anybody can read from and write arbitrary data. 
Since then, there has been a proliferation of such permissionless blockchains. They include general-purpose blockchains such as Ethereum, Cardano, Solana, Avalanche etc., each of which supports many decentralized applications, as well as application-specific blockchains such as the Cosmos zones, each of which supports a specific application. 
Together, these blockchains 
 form a {\em multi-chain world}, communicating with each other through bridging protocols, which read transactions from one ledger and write transactions onto another. These protocols provide important functionalities such as token swaps, thus allowing the composability of applications across different blockchains to build more powerful ones. 
Indeed, the Inter-Blockchain Communication (IBC) protocol,  which connects different Cosmos zones is a central reason for the recent flourishing of that ecosystem.

A fundamental property of a blockchain is its {\em security}. 
Currently, the security of a blockchain requires a majority or supermajority of its own validators to follow the protocol. 
So the security of a blockchain is as good as its own validator set.
In a multi-chain world, a natural question is whether a blockchain can borrow security from other blockchains through existing bridging protocols? 
In other words, can bridges be used to {\em transfer security} in addition to transfer assets, thus allowing the {\em composability of security}? 
Let us call such protocols {\em interchain consensus protocols}.

A concrete motivation for this question is the recent emergence of the {\em mesh security} concept for the Cosmos ecosystem \cite{sunny-mesh, mesh-sec-github}. This ecosystem currently consists of $54$ sovereign {\em zones}, each being a Tendermint-driven \cite{tendermint} Proof-of-Stake (PoS) blockchain. Since each such zone focuses on a specific application, individually, they tend to be smaller compared to general-purpose blockchains and hence have lower security. Thus, the ability to share security with neighboring zones is critically beneficial to both individual zones, as well as to the ecosystem as a whole since lack of security in one chain can spread to other chains. 

\subsection{Interchain Timestamping}

In this paper, we propose a simple candidate for PoS interchain consensus protocols, which we call the {\em interchain timestamping} protocol. 
The idea of timestamping or checkpointing to extract security from a chain has a long history and is used for example in finality gadgets \cite{ebbandflow,sankagiri_clc}, rollups and bridges. 
This idea has also been used recently in reducing latency of longest chain protocols \cite{ledger-combiners} and building PoS protocols enhanced by Bitcoin security \cite{bitcoin-timestamp, veriblock-whitepaper, komodo, pikachu, btc-pos}, where signed headers of PoS blocks are submitted as Bitcoin transactions to be timestamped on the Bitcoin chain. 
We apply this idea to the problem of one PoS \emph{consumer} blockchain extracting security from $k$ PoS \emph{provider} blockchains by having the block headers of the consumer chain timestamped on one of the provider chains, and the blocks of that provider chain timestamped on the next provider chain and so forth (Figure \ref{fig:ic_ts}).  

\begin{figure}[t]
    \centering
    \includegraphics[width=.9\linewidth]{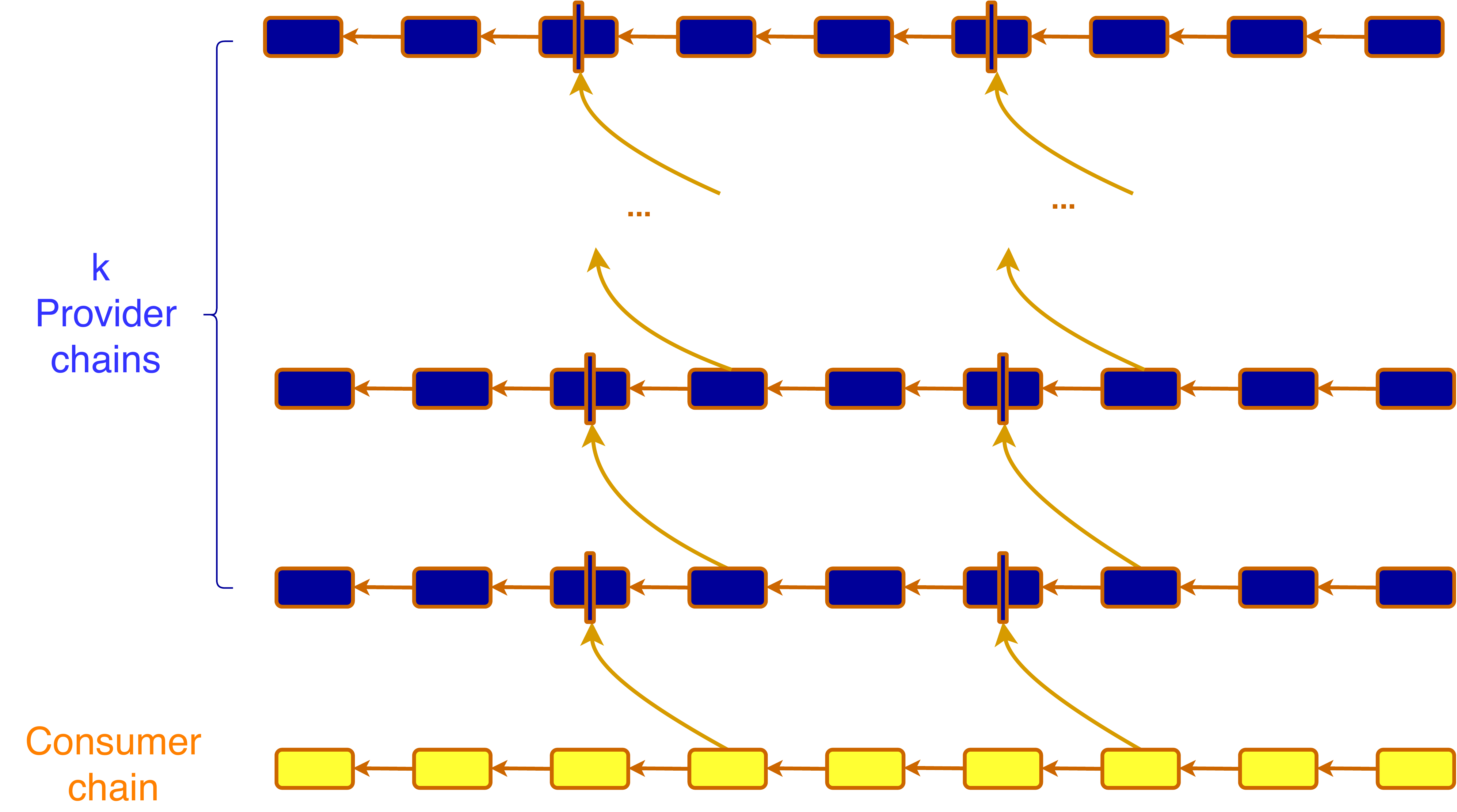}
    \caption{A consumer chain timestamping to $k$ provider chains in sequence.}
    \label{fig:ic_ts}
\end{figure}

How do we implement interchain timestamping using bridges? There are primarily two types of bridges, which connect a sender chain and a receiver chain: 1) multi-signature bridges, where bridge transfers are signed by a committee external to the validators of the blockchains; 2) light-client based bridges, where the validators of the sender chain sign block headers for transfers and the receiver chain maintains an on-chain light client of the sender chain to verify the transfers. IBC developed for Cosmos zones is a primary example of the latter. Perhaps a happy coincidence, but maintaining an on-chain light client of the sender chain in the receiver chain is nothing but putting the signed headers of the sender chain in the ledger of the receiver chain. Hence, light-client based bridges {\em automatically}  provide the timestamping functionality and no additional communication messages are needed. 
Figure~\ref{fig:ic_ts} gives an example of interchain timestamping in action. 

\begin{figure}
    \centering
    \includegraphics[width=1.0\linewidth]{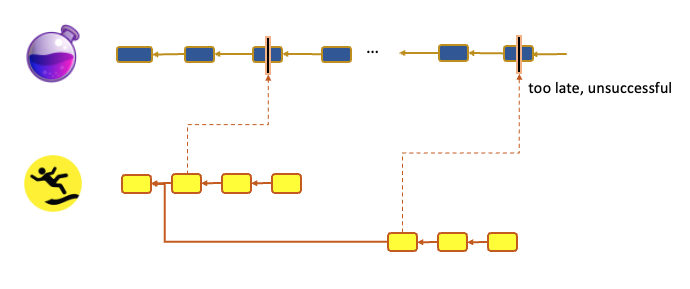}
    \caption{A new Cosmos zone "Rugpull" has an IBC channel with Osmosis, a Cosmos zone running a decentralized exchange. The headers of Rugpull are timestamped onto Osmosis through IBC. The timestamps help secure Rugpull, because if there is a safety attack to fork the Rugpull chain, the timestamps on Osmosis will determine the correct Rugpull fork. Securing Rugpull is also beneficial to Osmosis, since they are trading partners. }
    \label{fig:my_label}
\end{figure}
\subsection{Security Guarantees}

We provide a security analysis in the slashable safety framework \cite{casper, tendermint,forensics, acc_gadget}, where security of a consensus protocol is measured by two metrics: 1) how many validators can be held accountable and slashed when there is safety violation, the slashable safety resilience; 2) how many validators are needed to cause a liveness violation, the liveness resilience. 
In PoS protocols, validators stake funds in the protocol, and can be financially punished by burning, \ie, slashing their stake if they are found accountable for a safety violation.
The liveness resilience is a classical metric in consensus protocol theory, but the slashable safety resilience is a stronger concept than the classical safety resilience, since no assumption is made on how many validators are adversarial in the definition of slashable safety resilience. 
Since it directly translates to the cost of an attack, we will also call slashable safety resilience the {\em economic security} of the protocol. We will call the liveness resilience the {\em censorship resistance} of the protocol. Since censorship cannot always be provably attributed and slashed (\cf~\cite[Appendix F]{btc-pos}), this quantity does not directly translate into an economic security argument. In a classical partially synchronous protocol like Tendermint, both the economic security and the censorship resistance are $1/3$ of the total number of validators, translating into $1/3$ of the market cap of the blockchain as the funds staked in a PoS protocol typically constitute a large fraction of the market cap.

We obtain the following security results about the interchain timestamping protocol (timestamping protocol for short) in the slashable safety framework:
\begin{enumerate}
\item Interchain timestamping ensures the safety of the timestamped consumer chain as long as at least one of the consumer or provider chains is safe. Therefore, through interchain timestamping, the consumer chain can obtain {\em additional} economic security equal to the sum of the economic securities of the $k$ provider chains (Theorem~\ref{thm:interchain-timestamping-security}).
\item This economic security attained by interchain timestamping is the maximum among all possible interchain consensus protocols, which only read and write to the individual blockchains (Theorem~\ref{cor:interchain-timestamping-best-accountability}).
\item Interchain timestamping guarantees the liveness of the timestamped consumer chain if the consumer chain and all provider chains are live. Therefore, through interchain timestamping, the censorship resistance of the consumer chain becomes the minimum of the censorship resistances of the constituent chains.

\item Among all interchain protocols that only send succinct commitments from the consumer chain to the provider chain, interchain timestamping achieves optimal censorship resistance if the censorship resistance of the stand-alone consumer chain is less than the censorship resistance of every provider chain (Theorem~\ref{cor:interchain-timestamping-optimality-2}).
\end{enumerate}

Taken all together, these results say that under the reasonable assumption that the consumer chain has lower stand-alone security than each of the provider chains, the timestamping protocol is {\em simultaneously} optimal in both its slashable safety and in its liveness resilience guarantees, among all interchain consensus protocols which leave data on the consumer chain and only send commitments to the provider chains.

Although the protocol poses a trade-off between economic security and liveness resilience/latency in the general case, it gives clients of the consumer chain the \emph{flexibility} to select any provider chain without changing the behavior of the validators on these chains.
Hence, each client can determine (and later modify) its desired security-latency trade-off without coordinating with any other party, and any modification for the nodes who do not opt for enhanced economic security (Section~\ref{sec:safety-liveness-tradeoff}).

\subsection{Design Challenges}

Safety of the timestamping protocol depends \emph{only} on the safety resiliences of the constituent chains, whereas its liveness depends \emph{only} on their liveness resiliences. 
Although timestamping has been used extensively by multichain protocols (\eg Snap-and-Chat \cite{ebbandflow}, Babylon \cite{btc-pos}) for different goals, separation of safety and liveness conditions introduces unique challenges for the timestamping protocol, which limits the ability of existing protocols to address its security requirements.
For instance, a Snap-and-Chat implementation using light-client bridges would be susceptible to data availability attacks and need liveness of the constituent chains for safety.
In contrast, the Babylon protocol would stall and lose liveness after a safety violation in the consumer chain even if all chains are live.
To overcome these limitations, our protocol combines iterative timestamping on multiple providers with special \emph{stalling} and \emph{sanitization} rules.
Section~\ref{sec:related-work} analyzes how attacks against existing solutions inspired the design of the timestamping protocol.

\subsection{Interchain Timestamping for Mesh Security}

Interchain timestamping is a natural solution for mesh security, since the zones with IBC channels are already timestamping to each other.  We empirically evaluate the economic security gain in each Cosmos zone from interchain timestamping. The results are shown in the top of Figure \ref{fig:summary-zones}. The evaluation setup is described in Section \ref{sec:evals}. 
Empirical analysis shows that each Cosmos zone can derive economic security from any other zone at the cost of merely $6$s latency and $\$52,560$ per year.
We compare the results with that achieved by {\em cross-staking}, a technique proposed in \cite{sunny-mesh} to achieve mesh security. 
Cross-staking allows validators to stake their tokens not only on their native chain but also simultaneously on any other chain with which there is an IBC channel.
This means that validators are simultaneously downloading, validating and executing transactions on multiple zones as \emph{full nodes}, a significantly heavier form of security sharing than timestamping, where validators in each chain accept succinct commitments of blocks from other chains into their ledger but as \emph{light clients} do not download, validate or execute the data within these blocks on the other chains.
Similarly, the timestamping protocol allows clients of the consumer chain to run light clients of the provider chains with the same guarantees.
As such, timestamping is an example of an \emph{interchain consensus protocol}:
these protocols do not require any more interaction (\eg, running other consensus engines) among distinct validator sets than reading and writing to the ledgers of the other blockchains.

\begin{figure*}[t]
    \centering
    \begin{subfigure}{\textwidth}
        \centering
        \includegraphics[width=.9\textwidth]{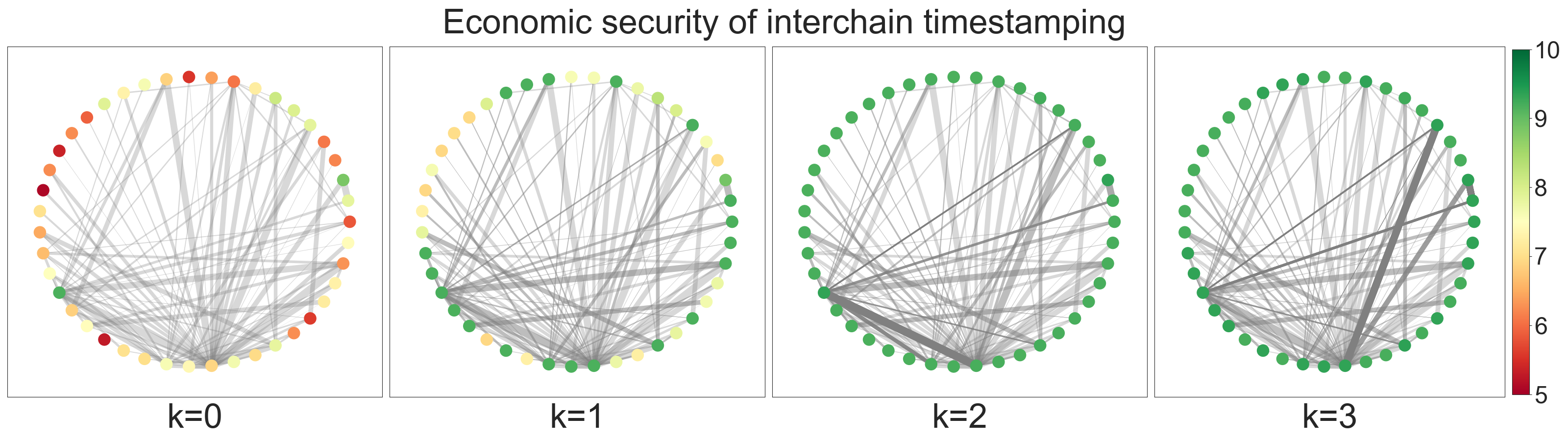}
    \end{subfigure}
    \begin{subfigure}{\textwidth}
        \centering
        \includegraphics[width=.9\textwidth]{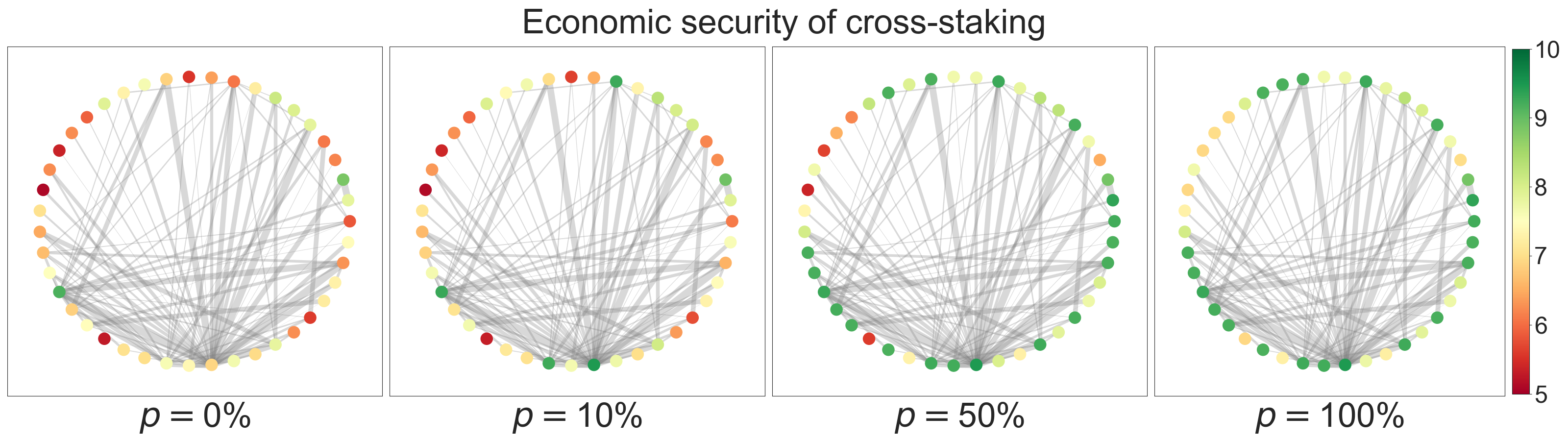}
    \end{subfigure}
    \caption{
        Economic security of interchain timestamping over Cosmos zones, with a comparison to cross-staking~\cite{cross-chain-validation-blog}.
        Each vertex represents a Cosmos zone and each edge represents an IBC channel.
        Vertex color is the maximum economic security in logscale: for example, a zone with color code $8$ achieves the economic security of $10^8$ USD.
        The edge thickness indicates the frequency of IBC transfers (thus light client updates), also in logscale.
        In interchain timestamping, parameter $k$ indicates the number of provider zones each consumer zone gets security from. 
        The edge transparency indicates whether the edge is used for providing security for a zone.
        In cross-staking, parameter $p$ in cross-staking indicates the power cap: the maximum percentage of the total stake on a consumer chain that can be cross-staked from a provider chain.
    }
    \label{fig:summary-zones}
\end{figure*}

\subsection{Outline}
Section~\ref{sec:related-work} compares the interchain timestamping prootocol with the related work, highlighting the technical challenges and novelty of our protocol with respect to older protocols that use timestamping.
Section~\ref{sec:preliminaries} presents definitions and the core assumptions.
Sections~\ref{sec:interchain-timestamping} and~\ref{sec:interchain-timestamping-security} provide a detailed description of the timestamping protocol and its security guarantees.
Section~\ref{sec:evals} gives an empirical evaluation of the security, latency and overhead of timestamping protocols for Cosmos zones.
In Section~\ref{sec:interchain-protocols}, we delineate the limits of the security properties achievable by interchain consensus protocols, characterize the settings where the timestamping protocol is optimal, and provide an achievability result for all optimal interchain consensus protocols.
Finally, in Section~\ref{sec:cross-chain-validation}, we characterize the limits of the security guarantees achievable by all possible cross-staking solutions and provide an achievability result for all optimal cross-staking protocols.

\section{Related work}
\label{sec:related-work}

\subsection{Snap-and-Chat Protocols}
\label{sec:related-work-snap-and-chat}

Snap-and-chat protocols post \emph{snapshots} of the ledger output by a permissioned longest chain (LC) protocol (\eg, Sleepy consensus~\cite{sleepy}), to a partially-synchronous BFT protocol (\eg, HotStuff~\cite{yin2018hotstuff}) to satisfy the \emph{ebb-and-flow property}, \ie to produce an available ledger, secure under dynamic participation, and a prefix ledger, secure under network partitions.
These snapshots consist of all transactions within the LC, thus making the protocol unpractical in terms of communication complexity.
Replacing them with succinct timestamps opens snap-and-chat protocols to \emph{data availability attacks}, where the timestamp of a confirmed LC (\ie, consumer) block appears in the (provider) chain output by the partially synchronous protocol, yet the block itself remains hidden from the clients. 
For instance, an adversarial majority among the LC validators can confirm a block $B$ and post its timestamp to the provider chain without revealing $B$ itself (this would not be possible if the entire LC data is posted as a snapshot).
Then, clients would have to skip the unavailable timestamp and create their ledgers using the available LC blocks.
However, the adversarial validators can later reveal $B$ to a late-coming client, which would include $B$ in its ledger.
As this is a safety violation despite the safety of the provider chain, snap-and-chat protocols fall short of the timestamping protocols, which satisfy safety when at least one chain is safe.

To avoid such safety violations, the timestamping protocol requires its clients to stall upon observing unavailable timestamps, thus \emph{trading-off} liveness for safety. 
Although the stalling rule introduces a new liveness attack vector, it does not reduce the liveness resilience of the timestamping protocol.

\subsection{Babylon}
\label{sec:related-work-babylon}

Babylon \cite{btc-pos} uses Bitcoin to protect PoS blockchains against long-range attacks and provide them with \emph{slashable safety}.
Towards this goal, it posts signed headers of PoS blocks as timestamps to the Bitcoin chain.
Unlike Babylon, our work assumes that long range attacks are resolved, either via \cite{btc-pos} or other means, and improves the economic security of PoS blockchains by timestamping their signed headers to multiple other PoS chains.

Although Babylon uses succinct timestamps and is equipped with the stalling rule, it cannot achieve the security objectives of the timestamping protocol even when Bitcoin is replaced with a PoS chain.
The timestamping protocol satisfies liveness when all constituent PoS chains are live, without any requirement on their safety.
In contrast, the chain of timestamped consumer blocks output by Babylon loses its liveness when there is a safety violation on the consumer chain.
This is because the fork-choice rule of Babylon requires the timestamped consumer blocks to be consistent with the blocks of earlier timestamps.
However, this cannot be satisfied when there is a safety violation on the consumer chain.

In the case of such safety violations, the timestamping protocol preserves liveness with a \emph{sanitization} step that resolves and merges consumer chain forks at different timestamps.
It thus ensures liveness when both the provider and consumer chains are live, without requiring their safety.

\subsection{Trustboost}
\label{sec:introduction-trustboost}

Trustboost~\cite{trustboost} is a class of consensus protocols, which treats multiple blockchains as `validators' of an external consensus protocol run on top of these blockchains (\cf~\cite{recursive-tendermint} for earlier discussions on such protocols).
The validator functionality is provided by a custom smart contract, and the emulated validators exchange messages over a cross-chain communication protocol (CCC).
Authors implement Trustboost by using smart contracts on the Cosmos zones and IBC as the CCC protocol.
Since the interaction among the validators is restricted to the IBC messages, which convey information from other zones' ledgers, Trustboost also represents an instance of the interchain consensus protocols.

Let us compare Trustboost and the timestamping protocol in terms of security properties.
Trustboost running a partially synchronous consensus protocol (\eg, HotStuff~\cite{yin2018hotstuff}, Tendermint~\cite{tendermint}, Streamlet~\cite{streamlet}) on top of the constituent blockchains satisfies security (safety and liveness) if over two-thirds of the constituent blockchains are secure~\cite{trustboost}.
For instance, Trustboost instantiated with two blockchains requires the security of both chains for the security of the Trustboost ledger (this is not a very interesting statement since a trivial protocol achieving the same guarantee is to simply use one of the constituent blockchains).
In contrast, the timestamping protocol requires the safety of only one of the constituent blockchains for the safety of the timestamped ledger.
However, unlike the timestamping protocol, Trustboost instantiated with $k>3$ blockchains does not require the liveness of all $k$ blockchains for liveness. 
As such, it is not comparable to the timestamping protocol in terms of optimality.
However, we believe that there is a potential path to unify the security guarantees of Trustboost and timestamping protocol in a single framework (\cf Section~\ref{sec:optimal-interchain}).

From an implementation point of view, Trustboost requires the deployment of custom smart contracts on the constituent blockchains and the exchange of specific messages beyond the timestamps provided by the existing IBC communication among Cosmos zones.
In contrast, given sufficient connectivity among the zones, timestamping does not require any changes or enhancements in the participating zones besides updating the clients of the zones to interpret timestamps.
Hence, the timestamping protocol maximizes the economic security achievable by any interchain protocol without any need to change or upgrade the underlying protocols.

\section{Preliminaries}
\label{sec:preliminaries}

\indent
\textbf{Notation.}
Let $[k]$ denote the set $\{0,1,2,\ldots,k\}$, and $\lambda$ denote the security parameter.
An event happens with \emph{negligible probability}, if its probability is $o(\frac{1}{\lambda^d})$ for all $d > 0$. 
An event happens with \emph{overwhelming probability} (w.o.p.) if it happens except with negligible probability.
We use `iff' as a short-hand for `if and only if'.

\textbf{Validators and clients.}
We consider two sets of protocol participants: validators and clients.
Validators receive transactions from the environment $\mathcal{Z}$, execute a state machine replication (SMR) protocol and send consensus-related messages to the clients.
Upon collecting messages from a sufficiently large quorum of validators, each client outputs a sequence of transactions called the \emph{ledger}.
In this case, the ledger is said to be \emph{finalized} in the client's view.
The validators' goal is to ensure that the clients output the same ledgers and obtain the same end states.
The set of clients includes honest validators, as well as external observers that can go offline arbitrarily and output ledgers at arbitrary times.
The protocol is permissioned and there is a public-key infrastructure (PKI): Validators have unique cryptographic identities, and their public keys are common knowledge.
In the following sections, we only consider \emph{static} validator sets, \ie, it is specified by the PKI and does not change throughout the execution.

\textbf{Blocks and chains.}
In \emph{blockchain protocols}, transactions are batched into \emph{blocks}, and the SMR protocol orders these blocks.
They are denoted by $B$, and consist of a header and transaction data.
The header contains (i) a pointer to a parent block (\eg hash of the parent block by a collision-resistant hash function), (ii) a binding and succinct vector commitment to the transaction data (\eg, a Merkle root), and (iii) consensus-related messages.
There is a genesis block $B_0$, that is common knowledge.

A block $B$ is a \emph{descendant} of $B'$ (respectively, block $B'$ is an ancestor of $B$), denoted by the prefix notation $B' \preceq B$, if $B'$ is the same as, or can be reached from $B$ by following the parent pointers.
A block is \emph{valid} iff it is a descendant of $B_0$, and all its ancestors have correct commitments to the respective transaction data in their headers. 
Thus, each valid block $B$ uniquely determines a \emph{chain}, denoted by $\chain$, starting at $B_0$ and ending at $B$.
Blocks $B$ and $B'$ \emph{conflict} if $B' \npreceq B$ and $B \npreceq B'$. 
In blockchain protocols, each client outputs a chain of blocks, from which a ledger can be extracted using the transaction data.

\textbf{Adversary.}
The adversary $\Adv$ is a PPT algorithm that corrupts a subset of the validators, hereafter called \emph{adversarial}, before the protocol execution commences.
Adversary takes control of these validators' internal states and can make them deviate from the protocol arbitrarily (Byzantine faults).
The remaining \emph{honest} validators faithfully follow the protocol rules.
We denote the number of adversarial validators by $f$ and the total number of validators by $n$.

\textbf{Networking.}
Time proceeds in discrete rounds and the clocks are synchronized\footnote{Bounded clock offset can be captured as part of the network delay.}.
Validators can send messages to each other and the clients through point-to-point channels, which are authenticated and reliable if there are honest validators or clients at the endpoints~\cite{lamport82}.
A validator is said to \emph{broadcast} a message if it is sent to all other validators and clients.
The adversary controls the schedule of message delivery and observes messages before the intended recipients.
Upon becoming online, a client observes all messages delivered to it while it was asleep.

We consider two network models:
In the \emph{synchronous} network, the adversary has to deliver a message sent by an honest validator to \emph{all} intended recipients within $\Delta$ rounds, where $\Delta$ is a known parameter.
Upon becoming online at round $t$, a client receives all messages sent to it by the honest validators before round $t-\Delta$.
In the \emph{partially synchronous} network~\cite{DLS88}, the adversary can delay messages arbitrarily until a global stabilization time ($\GST$) chosen by the adversary.
After $\GST$, the network \emph{becomes} synchronous, and the adversary has to deliver the messages sent by any honest validators to \emph{all} intended recipients within the known $\Delta$ delay.
Here, $\GST$ is unknown to the honest validators and clients, and it can be a causal function of the protocol randomness.
In the following sections, we assume that the network is partially synchronous unless stated otherwise.

\textbf{Security.}
Let $\chain^{\client}_r$ denote the chain output by a client $\client$ at round $r$.
We say that the protocol is secure with latency $\Tconfirm = \poly(\lambda)$ if:
\begin{itemize}
    \item \textbf{Safety:} For any rounds $r,r'$ and clients $\client,\client'$, either $\chain_{r}^{\client} \preceq \chain_{r'}^{\client'}$ or vice versa. 
    For any client $\client$, $\chain^{\client}_{r} \preceq \chain^{\client}_{r'}$ for all slots $r$ and $r' \geq r$.
    \item \textbf{$\mathbf{\Tconfirm}$-Liveness:} If $\mathcal{Z}$ inputs a transaction $\tx$ to an honest validator at some round $r$, then, $\tx \in \chain_{r'}^{\client}$ for all $r' \geq r+\Tconfirm$ and any client $\client$.
\end{itemize}

A protocol provides $\fS$-safety if it satisfies safety iff $f \leq \fS$, and $\fL$-$\Tconfirm$-liveness if it satisfies $\Tconfirm$-liveness iff $f \leq \fL$.

\textbf{Slashable safety.}
Each honest validator collects the exchanged consensus messages in an execution transcript.
If clients observe a safety violation, they send the conflicting chains as evidence to the validators, upon which the honest validators answer with their transcripts.
The clients then invoke a forensic protocol with these transcripts and generate a proof identifying $f_a$ adversarial validators as protocol violators~\cite{forensics}.
This proof is subsequently broadcast to all other clients, and serves as evidence of protocol violation.
\begin{definition}
\label{def:accountable-safety}
A blockchain protocol provides slashable safety with resilience $f_a$ if when there is a safety violation, (i) at least $f_a$ adversarial validators are identified by the forensic protocol as protocol violators, and (ii) no honest validator is identified (w.o.p.).
Such a protocol is said to provide \emph{$f_a$-slashable-safety}.
\end{definition}
As a stronger notion, $f_a$-slashable-safety implies $f_a$-safety. 
When safety is violated, $f_a$ adversarial validators are irrefutably identified; which cannot happen if fewer than $f_a$ validators are adversarial.

We hereafter assume a homogeneous stake distribution; since validators with more stake can be modeled as multiple unit-stake validators controlled by the same entity. 
Thus, identification of a fraction $\beta>0$ of adversarial validators imply the slashing of $\beta$ fraction of stake.
Moreover, as the total amount staked on a PoS blockchain is typically proportional to its market cap, slashable safety of a blockchain can thus be economically quantified as a fraction (\eg, $\beta$) of its market cap, allowing comparison of this economic security across different PoS chains.

\textbf{Data Availability.}
A valid block $B$ is \emph{available} in a client's view at round $r$ if the contents of $B$ and all its ancestors have been observed by the client by round $r$.
Upon outputting an ordering of the block headers, the clients can determine a total order across the transactions of available and valid blocks as the vector commitments in the headers are binding.
A validator executing a blockchain protocol is said to \emph{check for data availability} if it verifies the availability of the valid blocks before considering them valid.
In subsequent sections, all clients are \emph{full nodes} unless stated otherwise: They download all of the block headers and transaction data of the blockchain protocols.

\section{Interchain Timestamping Protocol}
\label{sec:interchain-timestamping}

We next describe the details of the timestamping protocol, and how it can be instantiated in the Cosmos ecosystem without changing the existing protocols.

\subsection{Timestamping on One Provider}
\label{sec:interchain-timestamping-two}

\begin{algorithm}[t]
    \captionsetup{font=small} 
    \caption{The function used by a bootstrapping client $\client$ to output the checkpointed ledger $\aux^\client_r$ at some round $r$. It takes the blocktrees $\T_C$ and $\T_P$ of finalized consumer and provider blocks as input, and outputs $\aux^{\client}_r$. The function $\textsc{GetCkpts}$ outputs the sequence of checkpoints on the unique provider chain extending the genesis provider block $B^P_0$. The function $\textsc{IsValid}$ checks if the given checkpoint contains pre-commit signatures by $q_C$ of the consumer validators on its block hash. The function $\textsc{GetChain}$ returns the chain of finalized consumer blocks within $\T_C$ that ends at the preimage of the hash within the checkpoint. It returns $\bot$ if the block or its prefix chain is unavailable or not finalized. The function $\clean(\aux,\chain)$ returns the sanitized sequence of blocks (\cf Figure~\ref{fig:i2qnterchain-timestamping}).}
    \label{alg.timestamping}
    \begin{algorithmic}[1]\small
    \Function{\sc OutputChain}{$\T_C, \T_P$}
        \Let{\ckpt_1, \ldots, \ckpt_m}{\textsc{GetCkpts}(\T_P)}
        \Let{\aux}{B^C_0}
        \For{$j=1$ to $m$} \Comment{Obtain the checkpointed ledger}
            \If{$\textsc{IsValid}(\ckpt_j)$} \Comment{Check validity.}
                \label{line:isvalid}
                \Let{\chain_j}{\textsc{GetChain}(\T_C, \ckpt_j)}
                \If{$\chain_j = \bot$}
                    \label{line:btc2}
                    \State\Return $\aux$ \Comment{Data Unavailable}
                \Else
                    \label{line:btc1}
                    \Let{\aux}{\clean(\aux,\chain_j)}
                    \label{line:update}
                    \Comment{Update chkpt. ledger}
                \EndIf
            \EndIf
        \EndFor
        \State\Return $\aux$ 
    \EndFunction
    \end{algorithmic}
\end{algorithm}

\begin{algorithm}[t]
    \captionsetup{font=small} 
    \caption{The timestamping protocol $\PI_I$ with $k+1$ blockchains. A client $c$ runs this protocol to determine the $\PI_I$ chain.}
    \label{alg.multitimestamping}
    \begin{algorithmic}[1]\small
    \Function{\sc $\PI_I$}{$\T_0, \dots, \T_k$}
    \For{$i = 0$ to $k-1$}
    \State $\T_P, \T_C \gets \T_{k-i}, \T_{k-i-1}$
    \If{$i < k-1$}
    \State $\T_{k-i-1} \gets {\textsc{OutputHeaderChain}(\T_C, \T_P)}$
    \Else
    \State $\T_{k-i-1} \gets {\textsc{OutputChain}(\T_C, \T_P)}$
    \EndIf
    \EndFor
    \State \Return $\T_{0}$
    \EndFunction
    \end{algorithmic}
\end{algorithm}

\begin{figure}[t]
    \centering
    \includegraphics[width=\linewidth]{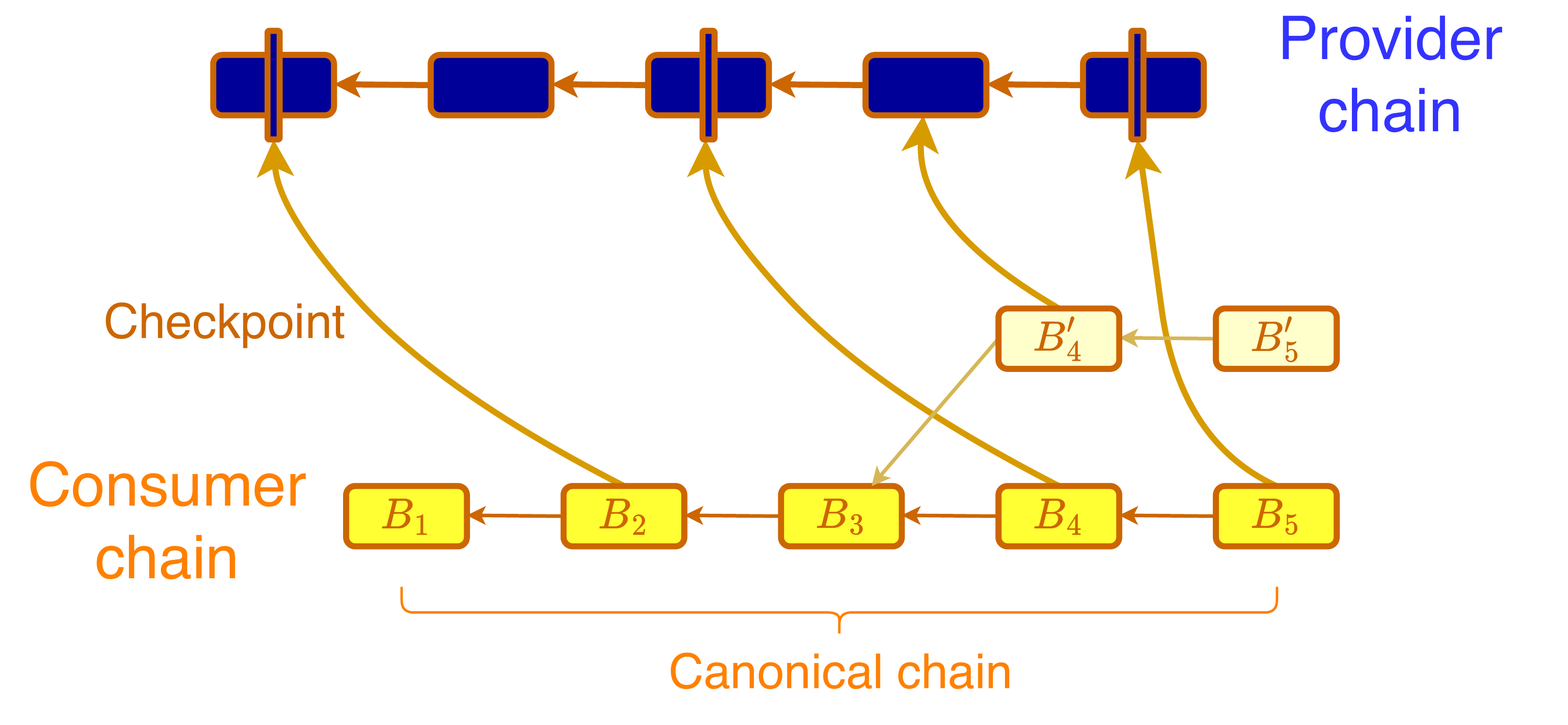}
    \caption{Interchain timestamping protocol. A client observing the consumer chain blocks $B_1$ thru $B_4$ outputs $B_1,B_2,B_3,B_4$ as the timestamped ledger ($\PI_I$ chain). After observing block $B'_4$, it updates the ledger by calling $\clean(B_1,B_2,B_3,B_4,\ B_1,B_2,B_3,B'_4)$, which after eliminating the duplicate blocks returns the sequence $B_1,B_2,B_3,B_4,B'_4$. 
    Upon observing the timestamp of $B_5$, the client outputs $\clean(B_1,B_2,B_3,B_4,B'_4\ B_1,B_2,B_3,B_4,B_5)$, which returns $B_1,B_2,B_3,B_4,B'_4,B_5$ as the timestamped ledger. Note that $B'_4$ does not have a pointer to $B_4$ in the ledger, thus in this example, it is not a chain.}
    \label{fig:i2qnterchain-timestamping}
\end{figure}

We first describe the timestamping protocol $\PI_I$ with two constituent blockchains $\PI_C$ and $\PI_P$.
We call $\PI_P$ receiving the timestamps the \emph{provider blockchain} (running the provider protocol) and $\PI_C$ at the origin of these timestamps the \emph{consumer blockchain} (running the consumer protocol).
This notation highlights the \emph{flow of security} enabled by the timestamping protocol.
We denote their blocks by $B_C$ and $B_P$, and call them the consumer and provider blocks respectively.
At a high level, validators of $\PI_C$ receive transactions from the environment and output the $\PI_C$ ledger.
Validators of $\PI_P$ receive timestamps of the $\PI_C$ ledger as their input transactions, and order these snapshots within the $\PI_P$ ledger.
At any round $r$, each online client $\client$ inspects the timestamps of the $\PI_C$ ledger in the $\PI_P$ ledger.
They extract the ledgers, whose timestamps were finalized by the $\PI_C$ validators, and output these ledgers in the order they appear in the $\PI_P$ ledger.
Finally, they eliminate the duplicate transactions appearing in multiple ledgers (Figure~\ref{fig:i2qnterchain-timestamping}) and output a timestamped ledger of consumer blocks as the $\PI_I$ ledger.

For constituent blockchains, we focus on quorum-based BFT protocols such as PBFT \cite{pbft}, Tendermint \cite{tendermint}, HotStuff \cite{yin2018hotstuff} and Streamlet \cite{streamlet} with slashable safety, and assume $\PI_C$ and $\PI_P$ are run by $n_C$ and $n_P$ validators with quorums of size $q_C$ and $q_P$ respectively.
These protocols enable clients to verify the finality of blocks by checking their quorum of signatures and allow validators to generate \emph{transferrable} and succinct proofs of finality consisting of these signatures (\cf certificate producing protocols \cite{lewispyeroughgardenccs}).
Therefore, finalized consumer (resp. provider) blocks hereafter refer to valid blocks, whose ancestors (including itself) have gathered a quorum of $q_C$ (resp. $q_P$) signatures from among the $n_C$ (resp. $n_P$) validators of $\PI_C$ (resp. $\PI_P$)\footnote{Exact name of the signatures depends on the protocol, it is called pre-commit in Tendermint \cite{tendermint} and commit in HotStuff \cite{yin2018hotstuff}).}.

\noindent
\textbf{Timestamps.}
The timestamps are succinct representations of finalized blocks that convey their finality.
They consist of the hash of the block (or blocks), its height, and a quorum of $q_C$ signatures on the hash.
When a new consumer block is finalized, an honest consumer validator (or a designated client) sends a timestamp to the provider blockchain as a transaction\footnote{Not every consumer block has to be timestamped and periodic timestamping is sufficient, albeit making latency larger for the $\PI_I$ ledger}.

\noindent
\textbf{Clients.}
We assume that each client $\client$ is a full node of the consumer chain and downloads the consumer blocks.
In contrast, from the provider chain, $\client$ has to obtain only the timestamps of the consumer chain.
Thus, it is sufficient for it to download the header chain of provider blocks and run a \emph{light client} of the provider chain akin to simple payment verification \cite{bitcoin}.
Indeed, $\client$ can succinctly verify the finality of the provider blocks, and receive all consumer chain timestamps within the finalized blocks from an honest full node (\eg validator) of the provider chain.
By using namespaced Merkle trees \cite{albassam2019lazyledger} to organize the provider chain data, these validators can convince $\client$ that it has received all consumer timestamps within a finalized provider block.

\noindent
\textbf{Sanitization.}
The sanitization function $\clean(\aux,\chain)$ takes a sequence of blocks $\aux$ (not necessarily a chain with consistent parent pointers) and a chain $\chain$ (\cf \cite{ebbandflow}).
It inspects their concatenation $\aux \mathbin\Vert \chain$, and eliminates the duplicate blocks that appear later (Figure~\ref{fig:i2qnterchain-timestamping}).
It outputs the remaining sequence of blocks.

\noindent
\textbf{Fork-choice rule (Alg.~\ref{alg.timestamping}).}
To identify the $\PI_I$ ledger at some round $r$, each client $\client$ first downloads all finalized consumer blocks and constructs a blocktree denoted by $\mathcal{T}_C$ with quorums of signatures on the blocks.
It also identifies the sequence $\ckpt_i$, $i \in [m]$, of consumer chain timestamps within the chain of finalized provider blocks, listed from the genesis to the tip of the provider header chain (if $\client$ observes a fork in the provider header chain, it outputs timestamps only from the portion preceding the fork).

Starting at the genesis consumer block, $\client$ constructs a \emph{timestamped} ledger of finalized consumer blocks, denoted by $\aux^{\client}_r$, by sequentially going through the timestamps.
This ledger is a sequence of blocks and imposes a total order across the consumer chain blocks and the transactions therein.
This total order is the same in the view of every client, and the clients can agree on a clean ledger by eliminating double-spends etc. in the same order.
However, the timestamped ledger need not be a chain, since a block in the ledger does not necessarily have a parent pointer to the previous block\footnote{If the consumer chain is safe, the timestamped ledger is a chain of consumer blocks.}.

For $i=1,\ldots,m$, let $\chain_i$ denote the chain of consumer blocks ending at the block, denoted by $B^C_i$, at the preimage of the hash within $\ckpt_i$, if $B^C_i$ and its prefix chain is available in $\client$'s view at round $r$.
Suppose $\client$ has gone through the sequence of timestamps until $\ckpt_j$ for some $j \in [m]$, and obtained $\aux$ as the latest timestamped ledger based on the blocktree $\mathcal{T}_C$ and $\ckpt_1 \ldots, \ckpt_j$.
The timestamp $\ckpt_{j+1}$ is said to be \emph{valid} if it contains $q_C$ pre-commit signatures by the consumer validators on its block hash. Then,

\noindent
\textbf{(1)} If (i) $\ckpt_{j+1}$ is valid, and (ii) every block in $\chain_{j+1}$ is available and finalized in $\client$'s view, then $\client$ sets $\aux \xleftarrow[]{} \clean(\aux,\chain_{j+1})$.

\noindent
\textbf{(2)} If (i) $\ckpt_{j+1}$ is valid, and (ii) a block in $\chain_{j+1}$ is either unavailable or not finalized in $\client$'s view, then $\client$ stops going through the sequence $\ckpt_j$, $j \in [m]$, and outputs $\aux$ as its final timestamped ledger.
This premature \emph{stalling} of the fork-choice rule is necessary to prevent data availability attacks.

After going through all timestamps or stalling early, $\client$ outputs the timestamped ledger $\aux^{\client}_r$ as the $\PI_I$ ledger.

\subsection{Timestamping on Multiple Providers}
\label{sec:interchain-timestamping-many}
We generalize the construction above by describing a timestamping protocol $\PI_I$ (Alg.~\ref{alg.multitimestamping}) with $k+1$ constituent blockchains $\PI_i$, $i \in [k]$, each with $n_i$ validators and a quorum of $q_i$ respectively.
Validators of $\PI_0$ receive transactions from the environment and output the $\PI_0$ chain.
For each $i=1, \ldots, k$, validators of $\PI_{i}$ receive timestamps of the finalized $\PI_{i-1}$ chain as their input transactions, and order them within their output chains.
The finalized $\PI_k$ chain is \emph{not} timestamped on any other chain.

At any round $r$, each online client $\client$ inspects the timestamps of the finalized $\PI_{i-1}$ chain in the $\PI_{i}$ chain for $i=1, \ldots, k$ ($\client$ is a light client of all blockchains except $\PI_0$).
It first determines the timestamped ledger of $\PI_{k-1}$ \emph{block headers} by using the protocol in Section~\ref{sec:interchain-timestamping-two} with $\PI_k$ acting as the provider blockchain and $\PI_{k-1}$ as the consumer blockchain.
Note that while running the protocol, $\client$ acts as a light client towards $\PI_{k-1}$ and obtains only the headers of the $\PI_{k-1}$ blocks constituting the timestamped $\PI_{k-1}$ ledger.
Then, for each $i=k-1, \ldots, 1$, by treating the timestamped ledger of $\PI_{i}$ headers as the finalized provider header chain and the finalized $\PI_{i-1}$ as the consumer blockchain, $\client$ iteratively repeats the protocol in Section~\ref{sec:interchain-timestamping-two} to output the timestamped ledger of $\PI_{i-1}$ headers.
In the final step, $\client$ acts as a full node towards $\PI_0$ and outputs the full timestamped ledger of $\PI_0$ blocks as the $\PI_I$ ledger.
In this architecture, all blockchains $\PI_i$, $i=1,\ldots,k$ act as providers, directly or \emph{indirectly}, towards $\PI_0$, which in turn acts as a consumer of all others.

\subsection{Interchain Timestamping via IBC}
\label{sec:ibc}

A promising real-world tool for realizing the interchain timestamping protocol is IBC (inter-blockchain communication).
It is the bonding agent of the Cosmos ecosystem, and has attracted attention from other PoS ecosystems.
It allows Cosmos zones, sovereign blockchains running specific applications, to send IBC packets (\eg, cross-chain transfers) to each other via an IBC channel.
A receiver zone verifies IBC packets from the sender zone by maintaining a light client of the sender zone.
The ordered IBC packets in the ledger are then executed at the application layer defined using Cosmos SDK~\cite{cosmos-sdk-source}.
The light client protocol of IBC has already implemented our desired timestamping mechanism.
Thus, utilizing it to extract security will not require any breaking changes to the existing systems.

Figure~\ref{fig:ibc_ts} depicts the light client protocol.
It consists of two main components:

\begin{figure}[t]
    \centering
    \includegraphics[width=1.0\linewidth]{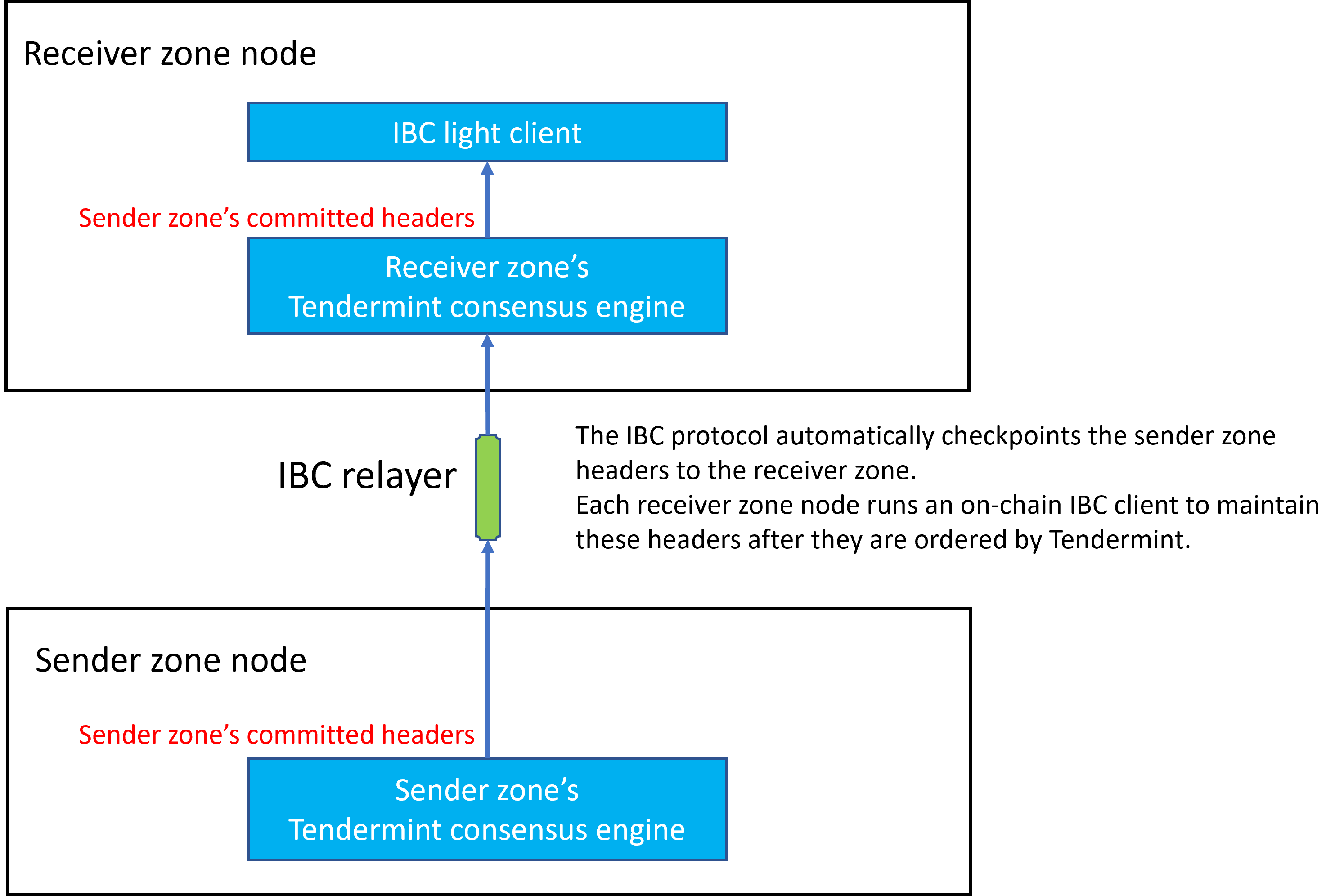}
    \caption{The IBC light client protocol does automatic timestamping.}
    \label{fig:ibc_ts}
\end{figure}

\noindent

\noindent
\textbf{(1)} An IBC relayer, which is an independent program connecting two validators from two zones. It relays headers and transactions of the sender zone to the receiver zone by sending them as the receiver zone’s transactions, so that the receiver zone's Tendermint engine can add them to the ledger.

\noindent
\textbf{(2)} An IBC light client, which is a Cosmos application layer module sitting inside each full node. It is responsible for maintaining the headers and verifying the transactions of the counter-party zone supplied by its Tendermint consensus engine.

Via the IBC light client, the receiver zone effectively already timestamps the sender zone's headers.
Since IBC is two-way, the two zones mutually timestamp each other and are ready to extract security from each other.
Implementing this security extraction only requires adding new APIs to the existing systems and running a standalone monitoring program. This monitor comprises a client (full node) of both the consumer and provider zones, and also a polling mechanism that:

\noindent
\textbf{(1)} periodically obtains the latest finalized consumer headers timestamped by the provider zone from its own provider zone client;

\noindent
\textbf{(2)} runs the sanitization function using these timestamped consumer zone headers with the consumer zone ledger obtained by its own consumer zone client. If they are conflicting, alerts with a proof of equivocation to slash the malicious consumer zone validators.

Thus, unless the provider zone is also compromised, the monitor always alerts if its consumer zone client is in a fork that is not timestamped by the provider zone. The consumer zone thus successfully extracts the security of the provider zone. 

Monitor can also be improved to only run light clients instead of full nodes. 
This requires an extra step in the polling, which is to verify the successful execution of the timestamped transactions in the provider zone\footnote{The current Tendermint implementation is a lazy consensus that follows the order-then-execute flow. This means a transaction's inclusion in the ledger does not mean successful execution and acceptance to the application state. }. This is a standard verification already supported by the Tendermint light client.
Akin to timestamping on multiple providers, one can extend the monitoring to any sequence of interconnected zones to extract their total economic security.

\section{Interchain Timestamping Security}

\label{sec:interchain-timestamping-security}
In this section, we state the security theorem for the interchain  timestamping protocol.
\begin{theorem}[Security of the Timestamping Protocol]
\label{thm:interchain-timestamping-security}
Given a set of adversarial validators, the timestamping protocol $\PI_I$ of Section~\ref{sec:interchain-timestamping-many} instantiated with $k+1$ blockchains $\PI_i$ with latency $T_i$, $i \in [k]$, under partial synchrony, satisfies
\begin{itemize}
    \item liveness with latency $\sum_{i=0}^k T_i$ (w.o.p.) for all PPT $\mathcal{A}$ iff \emph{all} of the constituent blockchains are live (w.o.p.) for all PPT $\mathcal{A}$.
    \item safety (w.o.p.) for all PPT $\mathcal{A}$ iff at least \emph{one} of the constituent blockchains is safe (w.o.p.) for all PPT $\mathcal{A}$.
    \item slashable safety (w.o.p.) such that there is a safety violation in $\PI_I$ iff safety is violated in \emph{all} of the constituent blockchains.
\end{itemize}
\end{theorem}
The proof is provided in Appendix~\ref{sec:appendix-timestamping-security}.
Since the transactions input to $\PI_0$ are directly or indirectly ordered by all of the constituent blockchains $\PI_i$, $i \in [k]$, $\PI_I$ remains safe as long as one of $\PI_i$, $i \in [k]$, satisfies safety and provides a consistent ordering.
In contrast, since all of the blockchains participate in the ordering, $\PI_I$ loses liveness as soon as one blockchain $\PI_i$, $i \in [k]$, loses liveness.
Note that the lack of liveness on one blockchain $\PI_j$ implies the lack of liveness for $\PI_I$ but not the lack of liveness for the blockchains $\PI_i$, $i \in [k], i \neq j$.

We next state a corollary of Theorem~\ref{thm:interchain-timestamping-security} for the timestamping protocol instantiated with $k+1$ blockchains, each running Tendermint with $n_i=3f_i+1$ validators and a quorum of $q_i = 2f_i+1$, $i=0,\ldots,k$.
It will be used in Section~\ref{sec:evals} when we analyze mesh security for Cosmos zones.
The corollary utilizes the security properties of Tendermint:
\begin{proposition}[Tendermint security, from~\cite{tendermint,tendermint_thesis}]
\label{prop:tendermint}
Tendermint with $n=3f+1$ validators and a quorum of $q=2f+1$ satisfies $f$-safety, $f$-liveness and $f+1$-slashable safety.
\end{proposition}
\begin{corollary}
\label{cor:interchain-timestamping-security}
The timestamping protocol $\PI_I$ of Section~\ref{sec:interchain-timestamping-many} instantiated with $k+1$ blockchains $\PI_i$, $i \in [k]$, each running Tendermint with $n_i = 3f_i+1$ distinct validators and a quorum of $q_i = 2f_i+1$ respectively, under partial synchrony, satisfies
\begin{itemize}
    \item liveness (w.o.p.) for all PPT $\mathcal{A}$ iff for \emph{all} blockchains $\PI_i$, $i \in [k]$, the number of adversarial $\PI_i$ validators is $f_i$ or less.
    \item safety (w.o.p.) for all PPT $\mathcal{A}$ iff for at least \emph{one} blockchain $\PI_i$, $i \in [k]$, the number of adversarial $\PI_i$ validators is $f_i$ or less.
    \item slashable safety (w.o.p.) such that if there is a safety violation, for all $i \in [k]$, $f_i+1$ adversarial validators are identified from the validator set of $\PI_i$, and no honest validator is identified (w.o.p.).
\end{itemize}
\end{corollary}

\subsection{Slashable Safety-Liveness/Latency Tradeoff}
\label{sec:safety-liveness-tradeoff}

Theorem~\ref{thm:interchain-timestamping-security} implies a trade-off between slashable safety, \ie economic security, and liveness resilience, \ie censorship resistance.
Using more provider chains to improve the economic security of the timestamped ($\PI_I$) ledger increases both the latency and the risk of losing liveness.
However, the timestamping protocol gives clients of each blockchain the ability to independently determine their trade-off. 
For illustration, consider $k+1$ blockchains $\PI_i$, $i=0,\ldots,k$, which are fully connected via IBC links.
Then, clients of $\PI_0$ can instantiate the timestamping protocol with any subset of the $k$ chains as their providers.
This decision affects only the clients' interpretation of the fork-choice rule as specified by Algs.~\ref{alg.timestamping} and~\ref{alg.multitimestamping}.
Hence, it does not require any change in the protocol code run by the validators of the constituent chains.

Flexibility of the timestamping protocol is not restricted to unanimous decisions by clients.
Different clients can use different telescopic sets of provider chains without sacrificing interoperability.
For instance, a conservative client $\client_1$ favoring safety can designate both $\PI_1$ and $\PI_2$ as its provider chains, whereas $\client_2$ favoring liveness might only use $\PI_1$.
While $\PI_0$ remains safe, \ie, under normal operation, $\client_1$'s timestamped ledger remains a prefix of $\client_2$'s ledger so that the clients agree on a single transaction history.

Flexibility of timestamping also allows conservative clients to resolve liveness violations by modifying the trade-off.
For instance, when $\client_1$'s timestamped ledger stops growing due to $\PI_2$ stalling, it can independently decide to stall or continue running the protocol with only $\PI_1$ after a timeout period.
In the latter case, it avoids stalling; however, the economic security of the new blocks will be proportional to the sum of the remaining provider chains.

\section{Mesh Security for Cosmos Zones}

\label{sec:evals}

As interchain timestamping is automatically provided by IBC, a natural application of our protocol is to provide mesh security to the Cosmos ecosystem. 
Here, we give an empirical evaluation of interchain timestamping for Cosmos zones.

\subsection{Goals of Evaluation}

The evaluation aims to answer the following questions:

\noindent
\textbf{Economic security upper bound:} Given the current market cap distribution and IBC channels of Cosmos zones, what is the economic security (\ie, slashable safety resilience) upper bound for Cosmos zones? (\ref{subsec:summary-zones-eval})

\noindent
\textbf{Security:} Given a consumer Cosmos zone and an upper bound on the number of provider zones, what is the best achievable economic security and the corresponding censorship resistance (\ie, liveness resilience) of the zone? (\ref{subsec:security-eval})

\noindent
\textbf{Latency:} What is the confirmation latency of the interchain timestamping protocol? (\ref{subsec:latency-eval})

\noindent
\textbf{Interchain timestamping v.s. cross-staking:} How does interchain timestamping compare with cross-staking~\cite{cross-chain-validation-blog} in terms of security? (\ref{subsec:comparison-cross-staking})

\textbf{Economic Security Upper Bound.}
Since the Cosmos zones form a mesh and each zone runs Tendermint with a quorum size equal to $2/3$ of the number of its validators, the economic security upper bound for each zone is $1/3$ of the total funds staked in all of the zones, and is achieved by the timestamping protocol instantiated with all of the zones (\cf Corollaries~\ref{cor:interchain-timestamping-security} and~\ref{cor:interchain-timestamping-best-accountability}).
As the total amount staked on each zone is proportional to its total market cap, the upper bound on economic security is proportional to the total market cap of all zones.
Given a consumer zone, a timestamping protocol using all other zones as providers can be constructed when there exists a directed IBC path that starts at the consumer zone and touches all other zones.
In the presence of a Hamiltonian path, all zones can act as consumer zones and extract the maximum economic security.
When zones have good connectivity, \eg, when the mesh of zones is fully connected, then every zone can achieve this upper bound.

\textbf{Security.}
The economic security of a Cosmos zone is proportional to the sum of the market caps of the zones on the path with the maximum total market cap among all paths going through this zone.
However, since censorship resistance is determined by the zone with the lowest market cap in a path, a long path is likely to reduce resilience to censorship.
In addition, if a path involves more zones, then the latency for the consumer zone increases, eventually becoming impractical for long paths.
Therefore, we quantify economic security for bounded path lengths $k$, and analyze the censorship resistance of the path that gives the best economic security among paths of bounded length.
For example, when $k=0$, both the economic security and censorship resistance of a zone depends only on its own market cap.
When $k=1$, the economic security of a zone is proportional to the sum of the market caps, while its censorship resistance is characterized by the market cap of the zone with the smaller market cap.

\textbf{Latency.}
The latency for the header of a consumer zone to be finalized in a provider zone consists of two parts: 
1) the latency for the header to be sent to the provider zone, and
2) the latency for finalizing a transaction in the provider zone.
The former depends on the frequency of updates for the IBC light clients, and the latter depends on the congestion level of the provider zone.
In turn, the frequency of updates in a connection depends on the IBC relayers.
The Cosmos community has noted that client updates are mostly triggered by pending IBC packets~\cite{client-update-frequency}, while the majority of IBC packets are IBC token transfers~\cite{ibc-token-transfer}.
Thus, we use the frequency of IBC token transfers as a proxy for the frequency of client updates.

The latency evaluation is two fold.
First, we evaluate the latency while assuming that relayers follow their current behaviours in our collected data.
Second, in order to show the trade-off between latency and the cost of interchain timestamping, we evaluate the latency and its corresponding cost w.r.t. the frequency of relaying headers.
The cost is quantified by the gas fee of transactions carrying \texttt{MsgUpdateClient}~\cite{msgupdateclient}, the message type that carries a header in the IBC protocol.

\textbf{Interchain Timestamping v.s. Cross-Staking.}
Another approach for implementing mesh security is cross-staking~\cite{cross-chain-validation-blog}.
In cross-staking, validators of a provider zone replicate their staked tokens to a consumer zone, such that these validators can participate in the consensus of the consumer zone.
A consumer zone $\Pi_i$ can set a parameter $q_j \in [0, 1]$ to limit the ratio between the value of stake from another zone $\Pi_j$ with an IBC channel and the total value of stake.
A blog post~\cite{cross-chain-validation-blog} from Informal Systems provides a model for analyzing the security of Cosmos zones with cross-staking.
In summary, given a consumer zone $\Pi_i$ with $X_i$ market cap from its native validators, it can borrow $X_{i,j}$ market cap from validators of other provider zones $\{\Pi_j\}_{j \in [k] \setminus i}$ with a direct IBC channel such that 
$\frac{X_{i,j}}{X_i+X_{i,j}} \leq q_j$, \ie, $X_{i,j} \leq \frac{q_j}{1-q_j}X_i$.
The economic security and censorship resistance of $\Pi_i$ are then at most $\frac{X_i + \sum_{j \in [k] \setminus i} X_{i,j}}{3}$.
A large $q_j$ allows $\Pi_i$ to borrow more security from $\Pi_j$, but weakens the sovereignty in the sense that validators from other zones have more voting power on this zone's consensus.

\subsection{Data Collection and Experimental Setting}

Answering these questions requires data of Cosmos zones, including (i) information of zones (e.g., chain ID), (ii) market caps of zones, and (iii) information of all IBC channels, (e.g., IBC transfers).

We obtained such data from \emph{Map of Zones}~\cite{map-of-zones}, an explorer for all Cosmos zones and their IBC channels.
Specifically, we retrieved the data by using the GraphQL~\cite{graphql} APIs provided by Map of Zones on Jan.\ 13, 2023. 
It is refreshed by Map of Zones every 24 hours.

After retrieving the data, we parsed it to a graph of all zones, where each vertex denotes a zone and each edge denotes the IBC channel between two zones (\cf Figure~\ref{fig:summary-zones}).
Each vertex carries the market cap of the corresponding zone as a property.
Each edge carries the frequency of IBC transfers in the corresponding IBC channel as a property.
Such a graph allows us to derive the security and latency of zones by using the above methodology, as well as the security of cross-staking by using the model in~\cite{cross-chain-validation-blog}.
We use Python's \texttt{NetworkX}~\cite{networkx} library for processing the data to such a graph.
The collected data and source code for processing it is available at~\cite{source-code}.

Increasing the frequency of client updates improves the latency but requires more transaction fees.
To quantify the trade-off between latency and monetary cost, we deploy a private Cosmos zone and connect it to multiple public Cosmos zones (Akash, Injective, Juno, Osmosis, Secret, Sei) via IBC relayer~\cite{ibc-relayer}, so that headers of these zones can be checkpointed to the private zone and the gas cost can be measured.
Note that the cost of checkpointing the header of different zones is the same because they all use the Tendermint-specific block header format.

\subsection{Summary of Cosmos Zones}
\label{subsec:summary-zones-eval}

The results show that there are 43 Cosmos zones with IBC activity and native tokens on their main net.
Their total market cap is about 9.1 billion USD, leading to the economic security upper bound of 3 billion USD when the zones are fully connected.

Figure~\ref{fig:summary-zones} depicts the graph of zones for different $k \in \{0, 1, 2, 3\}$.
In these figures, the vertex color indicates economic security, \ie, $1/3$ of the total market cap of the zone itself plus $k$ other zones as described above.
The color map is in logscale: if the color code is $8$, then the economic security is $10^8$ USD.
The edge thickness indicates the frequency of IBC transfers (thus light client updates), also in logscale.
The edge transparency indicates whether the edge is used for providing security for a zone.
If an edge is less transparent, then it is used by some zones for providing security.
Inspecting the color of the vertices, we observe that, with larger $k$, vertices become greener, meaning that zones obtain more security.
When $k = 3$, most zones obtain an economic security of $\geq 10^8$ USD except for some zones without any IBC channel.

From the transparency of edges, we find that zones choose a small subset of zones and channels to extract economic security.
These zones either have a high market cap, or connect to many other zones.
This is because a zone with higher market cap is more likely to provide better security, and a zone with better connectivity is more likely to connect to other zones with high market cap.

\subsection{Security Evaluation}
\label{subsec:security-eval}

\begin{figure*}[t]
    \centering
    \begin{subfigure}{.49\textwidth}
        \centering
        \includegraphics[width=\textwidth]{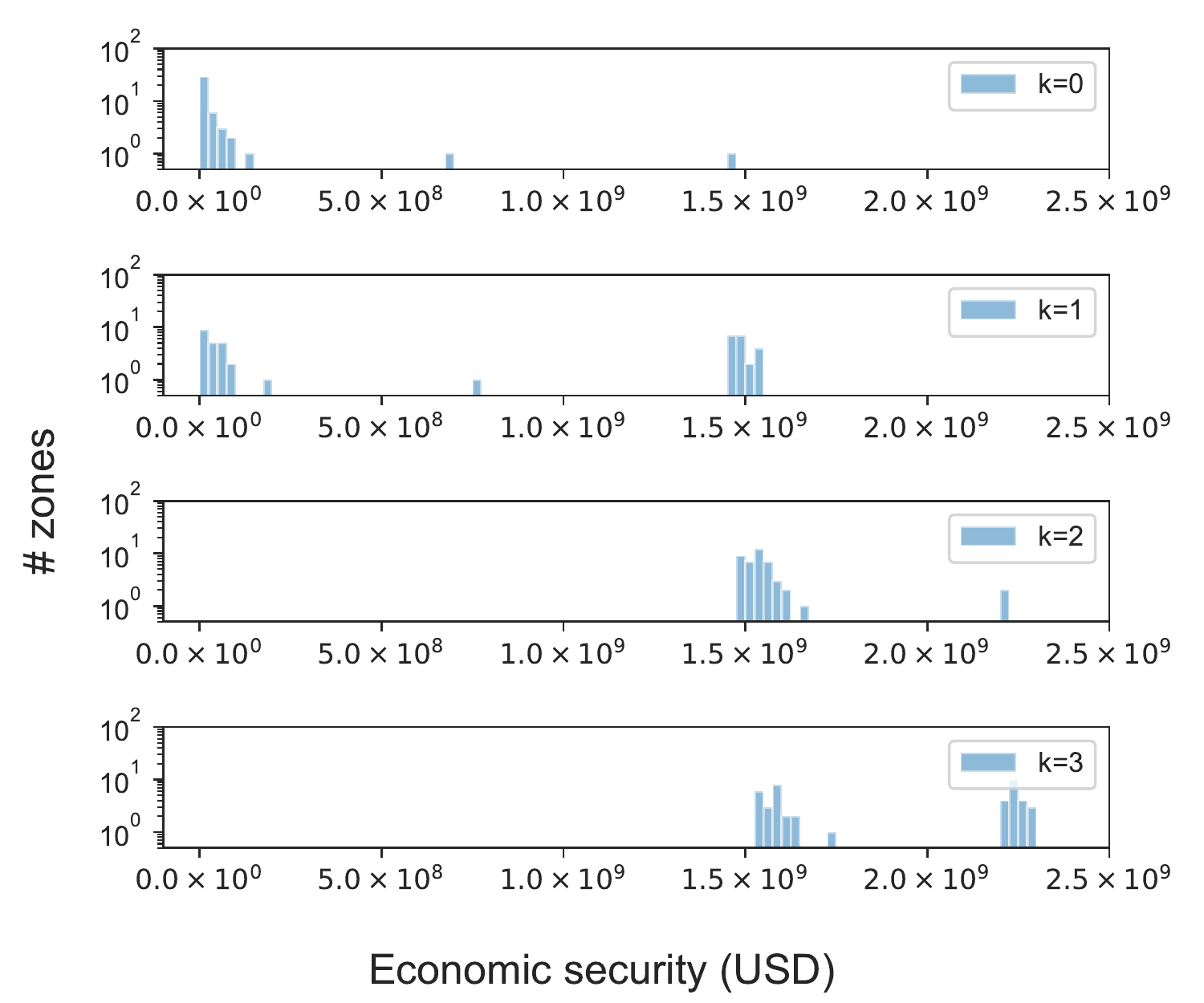}
        \label{fig:interchain-timestamping-economic-security}
    \end{subfigure}
    \begin{subfigure}{.49\textwidth}
        \centering
        \includegraphics[width=\textwidth]{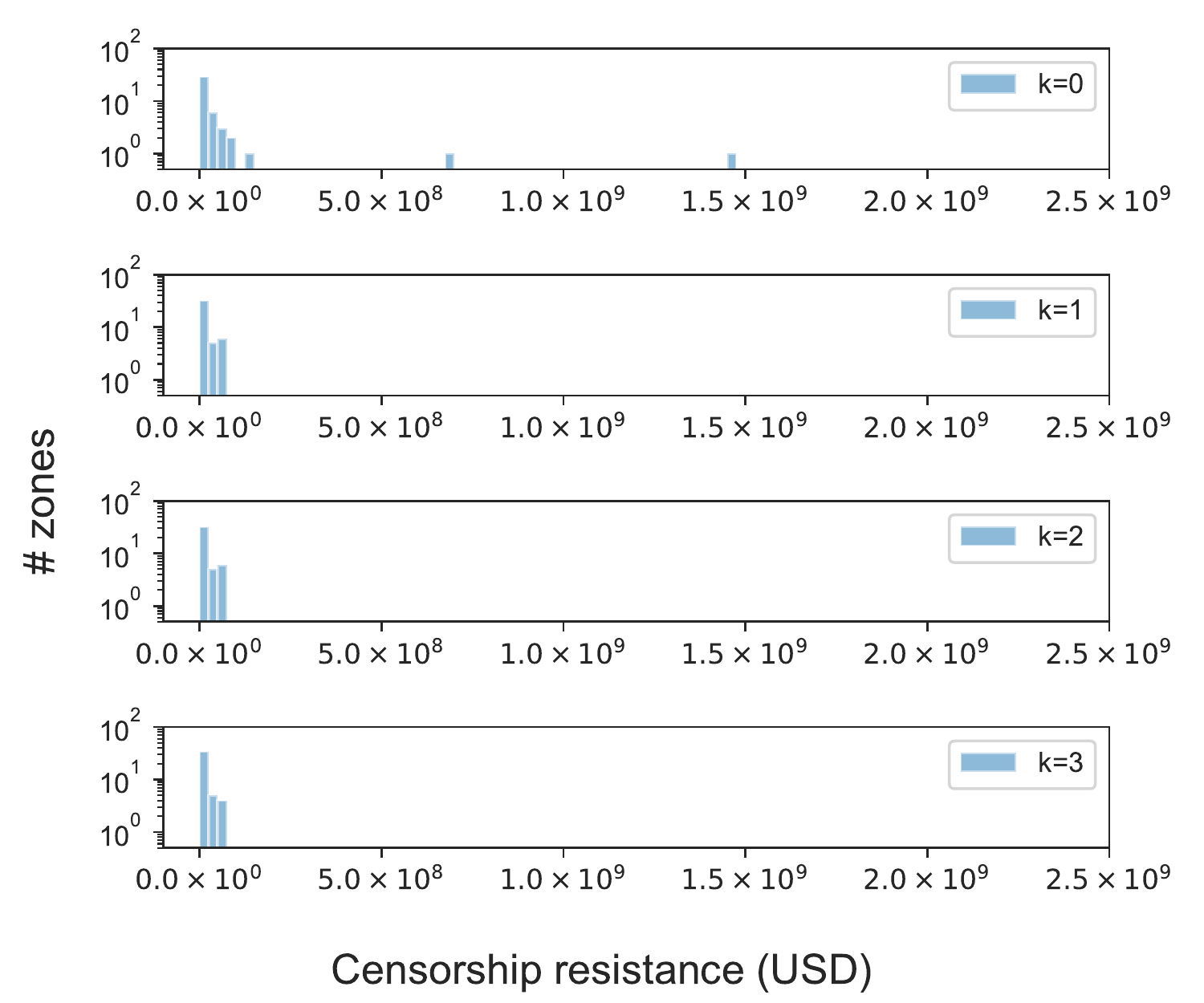}
        \label{fig:interchain-timestamping-censorship-resistance}
    \end{subfigure}
    \caption{Best achievable economic security of zones with various values for $k$, and the resulting censorship resistance.}
    \label{fig:security-hist}
\end{figure*}

Figure~\ref{fig:security-hist} provides the histogram of the best achievable economic security with different $k$ for Cosmos zones, as well as the censorship resistance resulting from paths that give them the best achievable economic security.
With larger $k$, zones achieve higher economic security except those without any IBC channels, similar to our observation in Figure~\ref{fig:summary-zones}.
In addition, we observe two gaps in the economic security, one is before $1.5$ billion USD and the other is between $1.5$ billion USD and $2.5$ billion USD.
The gap is because when $k$ becomes $2$ from $0$, zones with less than $0.2$ billion USD obtain more security from a specific zone with $0.7$ billion USD.
When $k$ becomes $3$ from $2$, some zones with the economic security of about $1.5$ billion USD obtain more security from that zone with $0.7$ billion USD as well.

Censorship resistance either remains stable or decreases with larger $k$ as it is determined by the zone with the lowest market cap in a path.
When $k$ increases by $1$, if the newly chosen zone has a higher market cap than any existing zone in the path, then the economic security remains stable; if it has a lower market cap than any existing zone, then the economic security decreases.
For example, when $k$ becomes $1$ from $0$, the censorship resistance of Cosmos Hub decreases to less than $0.2$ billion USD from $1.5$ billion USD.

\subsection{Latency and Cost Evaluation}
\label{subsec:latency-eval}

\begin{figure}[t]
    \centering
    \includegraphics[width=.85\linewidth]{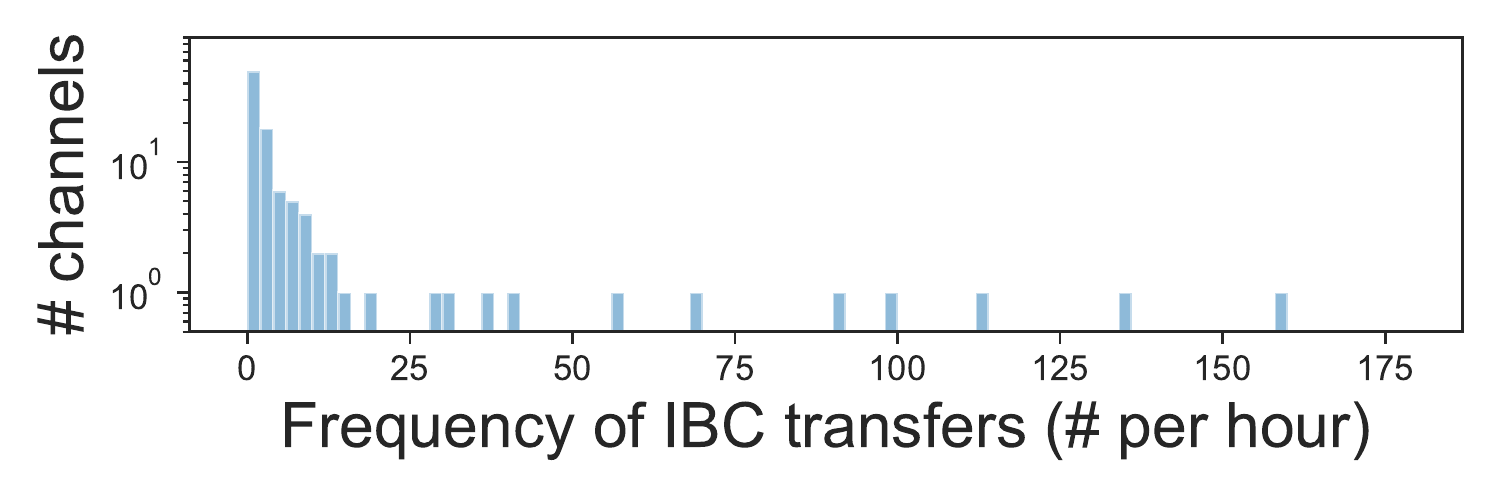}
    \caption{Frequency of IBC transfers for IBC channels.}
    \label{fig:sibc-transfer-frequency-hist}
\end{figure}

Figure~\ref{fig:sibc-transfer-frequency-hist} shows the frequency of IBC transfers for all IBC channels.
A small number of IBC channels are active, with up to 155 IBC transfers per hour, leading to at least $24$ seconds for relaying a header.
Meanwhile, most IBC channels have less than 25 IBC transfers (thus client updates) per hour, leading to at least $144$ seconds for relaying a header.
This is much longer compared to a zone's transaction confirmation latency, which is about no more than $6$ seconds in Cosmos Hub according to Mintscan~\cite{mintscan-cosmos-blocks}.
Therefore, when piggybacking the interchain timestamping to existing relayers without changing their behaviors, the latency for the provider zone to confirm a header of the consumer zone is $30$ seconds in the best case, and is at least 150 seconds when client updates are not often.

Zones can lower the latency by increasing the frequency of client updates.
This requires IBC relayers to pay more gas fees, leading to a trade-off between the latency and gas cost.
We evaluate the gas cost of a client update by running the relayer for about 12 hours, during which the relayer relayed 625 headers in total.
The results show that \emph{on average, each client update costs $217052.8 \pm 53472.1$ gases}, which are worth $0.01$ USD in Cosmos Hub according to the gas fee~\cite{mintscan-example-tx-client-update} and the price of ATOM~\cite{coinmarketcap} at the date of the experiment.

Therefore, with $0.01$ USD per header and the existing $6$-second transaction confirmation latency, relaying headers per $6$ seconds costs $144$ USD per day and $52560$ USD per year.
Assuming that the relayer's latency is also $x$ seconds per header, the worst-case latency of obtaining a timestamp for a header is $x+6$ seconds.
If the relayer's latency is $60$ seconds, then the worst-case latency is $66$ seconds, and the cost will be $5256$ USD per year.

We also observe that the standard deviation of the gas cost is large compared to the average value.
This is because verifying a header includes a quorum intersection, whose overhead largely depends on the size of the intersected signature set.

\subsection{Comparison to Cross-Staking}
\label{subsec:comparison-cross-staking}

Based on the model in~\cite{cross-chain-validation-blog}, we quantify the economic security upper bound for Cosmos zones with cross-staking in Figure~\ref{fig:summary-zones}, similar to the evaluation for Cosmos zones with interchain timestamping.
For every consumer zone, we set $q_j = p$ for every provider zone with an IBC channel, where we test $p = \{0\%, 10\%, 50\%, 100\%\}$.
Compared to interchain timestamping, vertices are less green in cross-staking, \ie, cross-staking provides a lower economic security upper bound than interchain timestamping, even with $p = 100\%$.
This is because cross-staking allows a zone to get security from only the zones with a direct IBC channel, bounding the total security obtainable from other zones.
Meanwhile, interchain timestamping allows a zone to get security from any zone reachable via a path of IBC channels.
To achieve better security in cross-staking, a zone needs to establish direct IBC channels with more zones.

In terms of latency, cross-staking allows validators to directly participate in the consensus of any zone with an IBC channel, so does not suffer from the latency of relaying headers or finalizing blocks as in interchain timestamping.

\subsection{Beyond Cosmos}
\label{subsec:beyond-cosmos}

Our interchain timestamping protocol can support light-client based bridges in other blockchain ecosystems besides Cosmos.
There have been efforts for bringing IBC functionalities to certain blockchains such as Ethereum (\eg, Polymer~\cite{polymer} and Electron~\cite{electron}), Near (\eg, Electron~\cite{electron}), Polkadot (\eg, Composable~\cite{composable}), and for building hubs supporting cross-chain communication (\eg, Axelar~\cite{axelar}).
The interchain timestamping protocol can be integrated with them to expand its adoption beyond Cosmos.
However, it cannot be directly applied to multi-signature bridges.
Thus, supporting those bridges remains as future work.

\section{Optimality of Interchain Timestamping}
\label{sec:interchain-protocols}

Among the class of interchain consensus protocols, is interchain timestamping optimal? In what sense is it optimal? We explore these questions in this section.

\subsection{Interchain Consensus Protocols}
\label{sec:interchain-protocols-definition}

An interchain consensus protocol (interchain protocol for short) $\PI_I$ is a SMR protocol executed using \emph{existing blockchains} with disjoint validator sets.
Its participants are the clients and validators of the constituent blockchains $\PI_i$, $i \in [k]$.
Clients and honest validators of each blockchain act as clients towards all other blockchains: for each $i$, validators of $\PI_i$ can read the output ledgers of $\PI_j$, $j \neq i$, and use the observed ledgers to determine the transactions to be input to $\PI_i$.
This communication among blockchains is captured by the \emph{cross-chain communication (CCC)} abstraction~\cite{ccc,trustboost} (\eg IBC, Section~\ref{sec:ibc}):
Each chain exposes its output ledger, whose \emph{finality} is verified by the clients and validators of the other chains; however, its internal mechanisms, \eg, validator set, is hidden, except as used by CCC to verify finality.
This encapsulation of the blockchains imposes limits on the properties of the interchain protocols.

\subsection{Quorum and Fail-Prone Systems}
\label{sec:quorum-failprone-definitions}
Before we analyze the interchain protocols, we introduce the notation needed to express their security properties.

\noindent
\textbf{Background.}
Consensus security is typically quantified by the maximum number $f$ of adversarial validators, for which the protocol remains secure.
For instance, Tendermint~\cite{tendermint} is secure for $3f+1 \leq n$.
To capture less uniform assumptions on the validators, we use the \emph{quorum and fail-prone systems} as defined by Malkhi and Reiter~\cite{malkhi98}.
Let $\UU$ denote the set of validators. 
A quorum system $\Q$ is a non-empty set of subsets of $\UU$, where each set $Q \in \Q$ is called a \emph{quorum}, and $\forall Q_1,Q_2 \in \Q \colon Q_1 \not\subset Q_2$.
Quorums represent smallest collections of validators that can drive consensus on behalf of the whole validator set.

A fail-prone system $\B$ is a non-empty set of subsets of $\UU$ such that $\forall B_1,B_2 \in \B \colon B_1 \not\subset B_2$.
Each set $B \in \B$ represents a potential set of adversarial validators under which the protocol remains secure.
\begin{definition}[Closure]
\label{def:closure}
Given a validator set $\UU$, we define the closure operation $\close(\Q)$ on a quorum system $\Q$ as: $\close(\Q) = \{S \subseteq \UU \colon \exists Q \in \Q, S \supseteq Q\}$.
We define the closure operation $\close(\B)$ on a fail-prone system $\B$ as: $\close(\B) = \{S \subseteq \UU \colon \exists B \in \B, S \subseteq B\}$.
\end{definition}
Intuitively, $\Q$ represents the set of \emph{smallest} quorums needed for liveness (\eg, sets of $2f+1$ validators), whereas $\close(\Q)$ gives the set of all possible quorums (\eg, sets with $2f+1$ or more validators).
Similarly, $\B$ represents the set of \emph{largest} adversarial validator sets (\eg, sets of $f$ validators), whereas $\close(\B)$ gives the set of all tolerable adversarial validator sets (\eg, sets with $f$ or less validators).

\noindent
\textbf{Fail-prone Systems for Liveness, Safety and Slashable Safety.}
To express the liveness, safety and slashable safety guarantees of interchain protocols, we extend the original definition for quorum and fail-prone systems (for a similar extension, \cf \cite{CachinT19}).
Let $F$ denote the set of adversarial validators.
A protocol $\PI$ is said to be $\B_s$-safe if $\PI$ is safe (w.o.p.) for all PPT $\mathcal{A}$ iff $\exists B_s \in \B_s \colon F \subseteq B_s$.
Namely, $\B_s$ is the set of validator sets such that the protocol is safe (w.o.p) for all PPT $\mathcal{A}$ iff the adversarial validators are covered by a set in $\B_s$.
A protocol $\PI$ is said to be $\B_a$-slashably-safe if whenever there is a safety violation, the validators in a set $B_a \in \B_a$, $B_a \subseteq F$, are identified by the forensic protocol as protocol violators, and no honest validator is identified (w.o.p.).

Since quorums have protocol-specific definitions, we redefine the quorum system $\Q$ as a security parameter applicable to all protocols.
A protocol $\PI$ is said to be $\Q$-live if $\PI$ is live (w.o.p.) for all PPT $\mathcal{A}$ iff $\exists Q \in \Q \colon F \cap Q = \emptyset$.
Intuitively, $Q$ represents sets of validators such that if the validators in $Q$ are honest, then the protocol is live (w.o.p.) for all PPT $\mathcal{A}$.
As an example, Tendermint with $n=4f+1$ validators and a quorum of $3f+1$ requires at least $3f+1$ honest validators for liveness, 
whereas it remains safe up to $2f$ adversarial validators.
Hence, it is $\Q$-live for $\Q$ containing every subset of $\UU$ with $3f+1$ validators, 
and $\B_s$-safe for $\B_s$ containing every subset with $2f$ validators.

We next define the class of \emph{pareto-optimal}\footnote{There are no games or rational actors in our model. We adopt the nomenclature of pareto-optimality to emphasize that a protocol dominating a pareto-optimal one in terms of one security property, \eg, safety, must necessarily have worse security in terms of the other properties, \eg, liveness.} interchain protocols that cannot be dominated in all dimensions of security by any other interchain protocol.
\begin{definition}[Dominating Points]
\label{def:dominating-points}
A tuple of quorum and fail-prone systems $(\Q,\B_s,\B_a)$ \emph{dominates} another tuple $(\Q',\B'_s,\B'_a)$ if $\close(Q) \supseteq \close(Q')$, $\close(\B_s) \supseteq \close(\B'_s)$, $\close(\B_a) \supseteq \close(\B'_a)$, and at least one of $\close(\Q)$, $\close(\B_s)$ or $\close(\B_a)$ is a strict superset of $\close(\Q')$, $\close(\B'_s)$ or $\close(\B'_a)$ respectively.
\end{definition}
\begin{definition}[Pareto-Optimal Interchain and SMR Protocols]
\label{def:protocol-optimality}
An interchain (SMR) protocol with the validator set $\UU$ and the quorum and fail-prone systems $Q$, $\B_s$ and $\B_a$ is \emph{pareto-optimal} under partial synchrony if there is no interchain (resp. SMR) protocol with the same validator set and the quorum and fail-prone systems $Q',\B'_s$ and $\B'_a$ such that $(Q',\B'_s,\B'_a)$ dominates $(Q,\B_s,\B_a)$.
\end{definition}
We denote the security guarantees of blockchains $\PI_i$, $i \in [k]$, that make up an interchain protocol $\PI_I$, by the quorum and fail-prone systems $\Q^i$, $\B^i_s$ and $\B^i_a$.
We denote the validator set, quorum and fail-prone systems of $\PI_I$ by $\UU^I$, $\Q^I$, $\B^I_s$, and $\B^I_a$.

\subsection{Upper Bounds on Interchain Protocols}
\label{sec:converses-interchain}
We now identify bounds on the quorum and fail-prone systems achievable by interchain protocols.
Consider an interchain protocol with the blockchains $\PI_i$, $i \in [k]$, validator sets $\UU^i$, quorum systems $\Q^i$ and fail-prone systems $\B^i_s$ and $\B^i_a$.
Given a set $Q' \subseteq \UU^I$, we define $f_Q \colon 2^{\UU^I} \to 2^{\UU^I}$,
\begin{IEEEeqnarray*}{C}
f_{Q}(Q') = \{i \in [k] \colon \exists Q \in \Q^i, Q' \cap \UU^i \supseteq Q\}
\end{IEEEeqnarray*}
as the function that returns the indices of the blockchains that are live (w.o.p.) for all PPT $\mathcal{A}$ given the set $Q'$ of honest validators.

Similarly, given a set $F \subseteq \UU^I$, we define the functions $f_s$ and $f_a \colon 2^{\UU^I} \to 2^{\UU^I}$,
\begin{IEEEeqnarray*}{rCl}
f_{s}(F) &=& \{i \in [k] \colon \forall B \in \B^i_{s}, F \cap \UU^i \not\subseteq B\} \\
f_{a}(F) &=& \{i \in [k] \colon \exists B \in \B^i_{a}, F \cap \UU^i \supseteq B\}
\end{IEEEeqnarray*}
The function $f_s(.)$ returns the indices of blockchains that are \emph{not} safe (with non-negligible probability) for some PPT $\mathcal{A}$ given a set $F$ of adversarial validators.
The function $f_a(.)$ returns the indices of the blockchains such that if the validators in $F$ are identified as adversarial validators after a safety violation for $\PI_I$, sufficiently many adversarial validators from these blockchains are also irrefutably identified as protocol violators.

For any given blockchain $\PI$, we assume that if $\forall Q \in \Q \colon F \cap Q \neq \emptyset$, then the adversarial validators can ensure (w.o.p.) that any desired client of $\PI$ outputs an empty $\PI$ ledger at all times.
\begin{theorem}[Safety-Liveness Trade-off under Partial Synchrony for Interchain Protocols]
\label{thm:safety-liveness-interchain-converse}
Consider an interchain protocol with the blockchains $\PI_i$, $i \in [k]$, validator sets $\UU^i$, quorum systems $\Q^i$ and fail-prone systems $\B^i_s$ under partial synchrony.
Then, it holds that $\forall Q^1, Q^2 \in \Q^I$ and $B_s \in \B^I_s \colon f_Q(Q^1) \cap f_Q(Q^2) \not\subseteq f_s(B_s)$.
\end{theorem}
Theorem~\ref{thm:safety-liveness-interchain-converse} generalizes the safety bound \cite[Theorem 4.4]{DLS88} on the tolerable adversary fraction under partial synchrony to the setting of interchain protocols.
Its proof follows from the techniques used in the proof of \cite[Theorem 4.4]{DLS88}.
\begin{theorem}[Slashable Safety-Liveness Trade-off for Interchain Protocols]
\label{thm:acc-safety-liveness-interchain-converse}
Consider an interchain protocol with the blockchains $\PI_i$, $i \in [k]$, validator sets $\UU^i$, quorum systems $\Q^i$ and fail-prone systems $\B^i_a$.
Then, it holds that $\forall Q^1, Q^2 \in \Q^I$ and $B_a \in \B^I_a \colon f_Q(Q^1) \cap f_Q(Q^2) \not\subset f_a(B_a)$.

Moreover, $\forall B_a \in \B_a, j \in [k] \colon \not\exists B^j_a \in \B^j_a, B_a \cap \UU^j \supset B^j_a$.
\end{theorem}
Theorem~\ref{thm:acc-safety-liveness-interchain-converse} generalizes the bound \cite[Theorem B.1]{forensics} on the number of adversarial validators that can be identified by a forensic protocol to the setting of interchain protocols.
Its proof follows from the techniques used in the proof of \cite[Theorem B.1]{forensics}.
The result holds under both synchrony and partial synchrony.

For interchain protocols $\PI_I$ outputting chains, when the validators of a constituent blockchain are not checking for the data availability of the output blocks, we further require each quorum of $\PI_I$ to fully cover a quorum from one of the blockchains whose validators check for data availability.
This ensures that the output chain of $\PI_I$ does not contain any unavailable block in the view of any client since validators checking data availability can prevent such blocks from being \emph{finalized}.
\begin{theorem}
\label{thm:data-limitation}
Consider an interchain protocol $\PI_I$ outputting a chain and executed using the blockchains $\PI_i$, $i \in [k]$, with validator sets $\UU^i$ and quorum systems $\Q^i$.
Suppose only the validators of $\PI_{i_j}$, $j \in [k']$, check the data availability of the blocks in the $\PI_I$ chain.
Then, $\forall Q \in \Q^I$, it holds that $\exists j \in [k'], i_j \in f_Q(Q)$.
\end{theorem}
Formal proofs of the Theorems~\ref{thm:safety-liveness-interchain-converse}, \ref{thm:acc-safety-liveness-interchain-converse} and \ref{thm:data-limitation} are given in Appendix~\ref{sec:appendix-interchain-proofs}.

We finally characterize a (potentially loose) upper-bound on the tuples of quorum and fail-prone systems of pareto-optimal interchain protocols (Definition~\ref{def:protocol-optimality}).
If achievable by any interchain protocol, tuples identified below correspond to the quorum and fail-prone systems of \emph{all} pareto-optimal interchain protocols.
\begin{definition}[Upper-Boundary of Quorum and Fail-Prone Systems]
\label{def:interchain-optimality}
A tuple $(\Q,\B_s,\B_a)$ of quorum and fail-prone systems is an \emph{upper-boundary point} under partial synchrony if they satisfy the conditions in Theorems~\ref{thm:safety-liveness-interchain-converse},\ref{thm:acc-safety-liveness-interchain-converse} and~\ref{thm:data-limitation}, and there is no tuple of quorum and fail-prone systems $(Q',\B'_s,\B'_a)$ dominating $(\Q,\B_s,\B_a)$ and satisfying the same theorems.
\end{definition}
\subsection{Property Based Security}
\label{sec:property-based-security}
In this section, we present a method to list all upper-boundary points, which will later be useful for arguing their \emph{achievability}.
Towards this goal, we define the property systems $\DQ$, $\D_s$ and $\D_a$ 
that map quorum and fail-prone systems of blockchains to assumptions on their safety and liveness, and thus constitute \emph{meta} quorum and fail-prone systems (\ie, property systems) for the interchain protocol executed on top of these chains.
These meta systems enable proving security of the interchain protocol as if it is an SMR protocol run by validators that correspond to the constituent chains.

Akin to $\Q$, $\B_s$ and $\B_a$, $\forall D_1,D_2 \in \DQ,\D_s,\D_a$, $D_1 \not\subset D_2$.
Given the adversarial validators, let $Q$ denote the set of constituent chains that satisfy liveness (w.o.p.) for all PPT $\mathcal{A}$.
Similarly, let $F_s$ denote the set of chains that do \emph{not} satisfy safety with non-negligible probability for some PPT $\mathcal{A}$.
Then, $\PI_I$ is said to be $\DQ$-live if it is live (w.o.p.) for all PPT $\mathcal{A}$ iff $\exists DQ \in \DQ \colon DQ \subseteq Q$.
Similarly, $\PI_I$ is said to be $\D_s$-safe if it is safe (w.o.p.) for all PPT $\mathcal{A}$ iff $\exists D_s \in \D_s \colon F_s \subseteq D_s$.
It is $\D_a$-slashably-safe if after every safety violation, $\exists D_a \in \D_a$ such that for all $i \in D_a$, adversarial validators in a set $B^i_a \in \B^i_a$ are irrefutably identified as protocol violators, and no honest validator is ever identified (w.o.p.).

The following theorem shows that the security properties of all pareto-optimal interchain protocols can be expressed by a tuple of property systems.
It uses generalizations of the functions $f_Q$, $f_s$ and $f_a$, now applied to sets of sets to output sets of quorums of blockchain indices.
\begin{IEEEeqnarray*}{rCl}
f_{Q}(\Q) &=& \{f_Q(Q) \colon Q \in \Q\} \\
f_{s}(\B_s) &=& \{f_s(B_s) \colon B_s \in \B_s\} \\
f_{a}(\B_a) &=& \{f_a(B_a) \colon B_a \in \B_a\}
\end{IEEEeqnarray*}
\begin{theorem}
\label{thm:interchain-quorum-optimality-2}
\noindent
\begin{enumerate}%
    \item No tuple of quorum and fail-prone systems achievable by an interchain protocol $\PI_I$ can dominate an upper-boundary point.
    \item If the quorum and fail-prone systems of an interchain protocol $\PI_I$ is an upper-boundary point, then $\PI_I$ is pareto-optimal. 
    \item The quorum and fail-prone systems $(\Q,\B_s,\B_a)$ of an interchain protocol $\PI_I$ is an upper-boundary point iff there exists a tuple of property systems $(\DQ, \D_s, \D_a \subseteq 2^{[k]})$ such that $\PI_I$ is $\DQ$-live, $\D_s$-safe, $\D_a$-slashably safe, and the property systems satisfy the following \emph{upper-boundary} conditions (\ie the tuple is an \emph{upper-boundary property point}):
    \begin{enumerate}
        \item $\forall DQ^1,DQ^2 \in \DQ, D_s \in \D_s \colon DQ^1 \cap DQ^2 \not\subseteq D_s$.
        \item $\forall DQ^1,DQ^2 \in \DQ, D_a \in \D_a \colon DQ^1 \cap DQ^2 \not\subset D_a$.
        \item Suppose only the validators of the protocols $\PI_{i_j}$, $j \in [k']$, check for the data availability of the $\PI_I$ blocks. Then, $\forall DQ \in \DQ$, it holds that $\exists j \in [k'] \colon i_j \in DQ$.
        \item No tuple of property systems satisfying (a)-(b)-(c) above dominates $(\DQ$, $\D_s$, $\D_a)$.
    \end{enumerate}
    Then, $\DQ = f_Q(\Q)$, $\D_s = f_s(\B_s)$ and $\D_a = f_a(\B_a)$.
\end{enumerate}
\end{theorem}
Proof of Theorem~\ref{thm:interchain-quorum-optimality-2} is stated in Appendix~\ref{sec:appendix-theorem-opt-2}.
It follows from the Definitions~\ref{def:protocol-optimality} and~\ref{def:interchain-optimality} and Theorems~\ref{thm:safety-liveness-interchain-converse},~\ref{thm:acc-safety-liveness-interchain-converse} and~\ref{thm:data-limitation}.

Theorem~\ref{thm:interchain-quorum-optimality-2} separates the quorum and fail-prone systems of interchain protocols achieving upper-boundary points into a tuple of property systems $(\DQ^I, \D^I_s, \D^I_a \subseteq 2^{[k]})$, satisfying the conditions (a)-(b)-(c)-(d), and the quorum and fail-prone systems of the constituent blockchains.
These conditions enable \emph{enumerating} all tuples of property systems constituting the upper-boundary property points for interchain protocols instantiated with a given collection of blockchains $\PI_i$, $i \in [k]$.
Hence, Theorem~\ref{thm:interchain-quorum-optimality-2} makes it possible to list all upper-boundary points for such interchain protocols through all upper-boundary property points.
However, it leaves open whether all upper-boundary points are achievable by an interchain protocol.

\subsection{Optimality of Timestamping}
\label{sec:is-timestamping-optimal}
\textbf{Timestamping is Optimal for Two Blockchains.}
To close the achievability gap, we first show that all upper-boundary points for interchain protocols $\PI_I$ with two blockchains $\PI_1$ and $\PI_2$ can be achieved by either trivial protocols or the timestamping protocol:

\begin{theorem}
\label{cor:interchain-timestamping-optimality}
All upper-boundary property points for interchain protocols instantiated with two blockchains are achievable by either trivial interchain protocols (\ie, empty ledger, using only one blockchain) or the timestamping protocol.
Hence, all quorum and fail-prone systems achievable by pareto-optimal interchain protocols with two blockchains can be achieved by either the trivial protocols or the timestamping protocol.
\end{theorem}
\begin{proof}[Proof of Theorem~\ref{cor:interchain-timestamping-optimality}]
For $k=2$, potential values for $\DQ$ are $\emptyset$, $\{\{0\}\}$, $\{\{1\}\}$, $\{\{0\},\{1\}\}$ and $\{\{0,1\}\}$.
When $\DQ = \{\{0\},\{1\}\}$, there is no property system $\D_s$, for which condition (a) can be satisfied, implying that no $\PI_I$ can be $\DQ$-live for this $\DQ$.
This leaves us with the sets $\DQ = \emptyset$, $\{\{0\}\}$, $\{\{1\}\}$, $\{\{0,1\}\}$.
Given these sets, the tuples $(\DQ,\D_s,\D_a)$ that satisfy conditions (a)-(b)-(c) and are dominated by no other tuple are given by $(\emptyset,\{\{0,1\}\},\bot)$, $(\{\{0\}\},\{\{1\}\},\{\{0\}\})$, $(\{\{1\}\},\{\{0\}\},\{\{1\}\})$ and more interestingly $(\{\{0,1\}\},\{\{0\},\{1\}\},\{\{0,1\}\})$.
Here, the first three tuples of property systems are respectively achieved by the following trivial protocols: empty ledger, $\PI_I = \PI_0$ and $\PI_I = \PI_1$.
The last one is the protocol that requires the liveness of both constituent chains for liveness, and remains safe as long as either one is safe.
By Theorem~\ref{thm:interchain-timestamping-security}, the timestamping protocol has exactly this security guarantee.

Finally, since all upper-boundary property points $(\DQ,\D_s,\D_a)$ for two blockchains are satisfied by the interchain timestamping protocol (or trivial protocols), by Theorem~\ref{thm:interchain-quorum-optimality-2}, all upper-boundary points for two blockchains can be achieved by either the timestamping protocol or trivial protocols.
Then, again by Theorem~\ref{thm:interchain-quorum-optimality-2}, quorum and fail-prone systems of all pareto-optimal interchain protocols with two blockchains are upper-boundary points and can be achieved by either trivial interchain protocols or the timestamping protocol.
\end{proof}

\noindent
\textbf{Timestamping Has the Best Slashable Safety Guarantee.}
The timestamping protocol achieves the strongest slashable safety guarantee among interchain protocols with the best possible liveness resilience:
\begin{theorem}
\label{cor:interchain-timestamping-best-accountability}
For any interchain protocol $\PI_I$ with the blockchains $\PI_i$, $i \in [k]$, and the quorum and fail-prone systems $\Q$ and $\B_a$, the quorum and fail-prone systems $\tilde{\Q}$ and $\tilde{\B}_a$ of the timestamping protocol $\tilde{\PI}_I$ with the same blockchains satisfies $\close(\tilde{\B}_a) \supseteq \close(\B_a)$.
Moreover, the timestamping protocol is pareto-optimal.
\end{theorem}
\begin{proof}[Proof of Theorem~\ref{cor:interchain-timestamping-best-accountability}]
Since the tuple of property systems of the timestamping protocol is an upper-boundary property point, by Theorem~\ref{thm:interchain-quorum-optimality-2}, the timestamping protocol is pareto-optimal.
By Theorem~\ref{thm:interchain-timestamping-security}, there is a safety violation in $\tilde{\PI}_I$ iff safety is violated in all of the constituent blockchains.
Therefore, $\tilde{\D}_a = [k]$, which is the largest property system for slashable safety among all upper-boundary property points.
Then, by Theorem~\ref{thm:interchain-quorum-optimality-2}, $\tilde{D}_a = f_a(\tilde{B}_a)$ and the protocol also has the largest fail-prone system $\tilde{\B}_a$ for slashable safety among all upper-boundary points. 
This implies $\close(\tilde{\B}_a) \supseteq \close(\B_a)$ for all interchain protocols $\PI_I$.
\end{proof}

\noindent
\textbf{Optimality of Timestamping for Multiple Blockchains.}
We next highlight a situation with multiple blockchains that is commonly observed within multichain ecosystems like Cosmos.
Suppose only the $\PI_0$ validators check for the data availability of the $\PI_I$ blocks, and the cost of corrupting validators is the same for the adversary $\mathcal{A}$ across all blockchains, with a total budget of $f$ adversarial validators.
We then replace the quorum and fail-prone systems of $\PI_I$ with numbers $f_\ell$ and $f_s$ such that $\PI_I$ is live and safe (w.o.p.) for all PPT $\mathcal{A}$ iff $f<f_\ell$ and $f<f_s$ respectively.
Here, $f_a$ is defined as the number of adversarial validators identified across all chains in the event of a safety violation.
This formulation is essentially a projection of the fail-prone systems onto a line, where each system is mapped to the cardinality of the smallest set in it.
Consider pareto-optimal interchain protocols $\PI_I$ instantiated with the blockchains $\PI_i$, $i \in [k]$, each running Tendermint with $n_i = 3f_i+1$ validators and a quorum of $q_i=2f_i+1$ such that $f_0 \leq f_i$ for all $i \in [k]$.
In this case, we observe that all quorum and fail-prone systems of non-trivial pareto-optimal interchain protocols are achieved by the timestamping protocol:
\begin{theorem}
\label{cor:interchain-timestamping-optimality-2}
Consider interchain protocols $\PI_I$ instantiated with the blockchains $\PI_i$, $i \in [k]$, each running Tendermint with $n_i = 3f_i+1$ validators and a quorum of $q_i=2f_i+1$ such that $f_0 \leq f_i$ for all $i \in [k]$.
Suppose only the $\PI_0$ validators check for the data availability of the $\PI_I$ blocks.
Then, no interchain protocol can achieve a liveness resilience larger than $f_0$.
Similarly, no interchain protocol can achieve a slashable safety resilience larger than $k+\sum_{i \in [k]} f_i$.
The timestamping protocol instantiated with these $k+1$ blockchains achieves $f_0$-liveness and $k+\sum_{i \in [k]} f_i$-slashable-safety.
\end{theorem}
\begin{proof}[Proof of Theorem~\ref{cor:interchain-timestamping-optimality-2}]
Since only the $\PI_0$ validators check for data availability, by condition (c) of Theorem~\ref{thm:interchain-quorum-optimality-2}, $\PI_I$ is live only if $\PI_0$ is live, \ie only if $f \leq f_0$.
Moreover, for any given $i \in [k]$, $\PI_i$ is live (w.o.p., for all PPT $\mathcal{A}$) iff $f \leq f_i$.
Hence, as $f_0 \leq f_i$ for all $i \in [k]$, liveness of $\PI_I$ implies liveness of $\PI_i$ for all $i \in [k]$ (w.o.p., for all PPT $\mathcal{A}$), and $\PI_I$ cannot achieve a larger liveness resilience than $f_0$. 
By Theorem~\ref{thm:interchain-quorum-optimality-2}, the only such upper-boundary property point is $\DQ = \{[k]\}$, $\D_s = \{\{0\}, \{1\}, \ldots, \{k\}\}$ and $\D_a = \{[k]\}$.
We know from Theorem~\ref{thm:interchain-timestamping-security} that the timestamping protocol instantiated with these $k+1$ blockchains satisfies this tuple of property systems, and achieves $f_0$-liveness, $\sum_{i \in [k]} f_i$-safety, and $k+\sum_{i \in [k]} f_i$-slashable-safety.
\end{proof}

\subsection{Closing the Achievability Gap}
\label{sec:optimal-interchain}
Unlike the case with two blockchains, there are upper-boundary points for three or more blockchains such that their quorum and fail-prone systems cannot be achieved by the timestamping protocol instantiated with these blockchains.
Can these points be achieved by any other interchain consensus protocol such as Trustboost~\cite{trustboost}?
For instance, in the case of three blockchains, the interchain timestamping protocol is secure iff all three constituent blockchains are live and at least one of them is safe. Trustboost, in turn, is secure only if over two-third of the constituent blockchains are secure. Therefore, in the case of three blockchains, Trustboost is secure iff all constituent blockchains are both safe and live. 
Thus, neither the interchain timestamping protocol nor Trustboost can achieve the tuple of property systems $\DQ = \{\{0,1\}, \{0,2\}, \{1,2\}\}$, $\D_s = \emptyset$, $\D_a = \{\{0\},\{1\},\{2\}\}$ (an upper-boundary property point), \ie, the protocol is live iff at least two of the constituent blockchains are live, and safe iff all blockchains are safe (for all adversaries $\mathcal{A}$, w.o.p.).
This leaves us with an achievability gap.

Recall that in Trustboost, each constituent blockchain emulates a validator of a partially synchronous consensus protocol run on top of these blockchains (this protocol is not necessarily the same as the protocol executed by the validators of each individual blockchain). 
We argue that the achievability gap can be closed by changing the quorum of the partially synchronous consensus protocol used in Trustboost. 
More specifically, we can replace the threshold rule with the quorum systems. 
For instance, suppose Trustboost uses HotStuff, and the quorum certificate of HotStuff requires votes from over two-thirds of the constituent blockchains.
Then, to achieve the upper-boundary point $(\Q,\B_s,\B_a)$, we remove the threshold rule and allow any set of blockchains within $\DQ = f_Q(\Q)$ to act as a quorum certificate of HotStuff.
In the above example, it would be any two chains out of the three chains.

We conjecture that the modified protocol achieves any desired upper-boundary point.
A smart contract on a constituent blockchain without safety behaves like an equivocating validator, whereas a smart contract on a blockchain without liveness behaves like a validator with omission faults. Therefore, if all three blockchains are safe, then none of them equivocates, so HotStuff is safe, implying the safety of the Trustboost ledger. If two blockchains are live, HotStuff can always create a quorum certificate of two votes, so HotStuff is live, implying the liveness of the Trustboost ledger. This argument can be generalized for any $k > 2$:
\begin{theorem}
\label{thm:optimality-of-trustboost}
For any positive integer $k$, all upper-boundary points for interchain protocols instantiated with $k$ blockchains can be achieved by a Trustboost protocol instantiated with HotStuff and the same blockchains.
Hence, all quorum and fail-prone systems achievable by pareto-optimal interchain protocols with $k$ blockchains can be achieved by a Trustboost protocol instantiated with HotStuff and the same blockchains.
\end{theorem}
Proof of Theorem~\ref{thm:optimality-of-trustboost} is given in Appendix~\ref{sec:appendix-interchain-proofs}.

The complexity of Trustboost highlighted in Section~\ref{sec:introduction-trustboost} raises the question whether there is an interchain protocol that retains the simplicity of timestamping, yet achieves all upper-boundary points.

\section{Cross-staking}
\label{sec:cross-chain-validation}

Finally, we analyze general cross-staking solutions, where validators are not restricted to the quorum systems of their blockchains as in interchain protocols, but can form arbitrary quorum systems.

\subsection{Optimality for SMR Protocols}
\label{sec:converses-general}
We first identify the limits of quorum and fail-prone systems achievable by any SMR protocol under partial synchrony.
\begin{theorem}[Safety-Liveness Trade-off under Partial Synchrony]
\label{thm:psync-safety-liveness-converse}
For every SMR protocol $\PI$ that is $\Q$-live and $\B_s$-safe under partial synchrony, it holds that $\forall Q^1, Q^2 \in \Q$ and $B_s \in \B_s \colon Q^1 \cap Q^2 \not\subseteq B_s$.
\end{theorem}
\begin{theorem}[Slashable Safety-Liveness Trade-off]
\label{thm:acc-safety-liveness-converse}
For every SMR protocol $\PI$ that is $\Q$-live and $\B_a$-slashably-safe, it holds that $\forall Q^1, Q^2 \in \Q$ and $B_a \in \B_a \colon Q^1 \cap Q^2 \not\subset B_a$.
\end{theorem}
Theorem~\ref{thm:psync-safety-liveness-converse} applies the techniques used in the proof of \cite[Theorem 4.4]{DLS88} to the setting of quorum and fail-prone systems.
Theorem~\ref{thm:acc-safety-liveness-converse} applies the techniques used in the proof of \cite[Theorem B.1]{forensics} to the setting of quorum and fail-prone systems.
The result holds under both synchrony and partial synchrony.

When all validators are not checking for the data availability of the blocks, we further require the following condition for the quorum system.
\begin{theorem}
\label{thm:general-data-limitation}
Consider an SMR protocol with a quorum system $\Q$.
Suppose only the validators in some set $\UU' \subseteq \UU$ check for the data availability of the output blocks.
Then, for all $Q \in \Q$, it holds that $Q \cap \UU' \neq \emptyset$.
\end{theorem}
Proofs of the Theorems~\ref{thm:psync-safety-liveness-converse},\ref{thm:acc-safety-liveness-converse} and \ref{thm:general-data-limitation} are given in Appendix~\ref{sec:appendix-general-proofs}.
\subsection{Closing the Achievability Gap}
\label{sec:achievability-general}
\begin{theorem}
\label{thm:hotstuff}
Any tuple of quorum and fail-prone systems $(\Q$, $\B_s$, $\B_a)$ achievable by a pareto-optimal protocol under partial synchrony can be achieved by HotStuff executed with the quorums in $\Q$.
\end{theorem}
Theorem~\ref{thm:hotstuff} 
give a construction achieving all pareto-optimal SMR protocols under partial synchrony.
Its proof is the \emph{same} as the security proof for HotStuff~\cite{yin2018hotstuff} 
except that the quorums are changed to be the sets in $\Q$.
General results under a synchronous network can be found in Appendix~\ref{sec:appendix-smr}.

\subsection{Gap between the Interchain and SMR Protocols}
\label{sec:gap-interchain-smr}

The pareto-optimal interchain protocols constitute a strict subset of the pareto-optimal SMR protocols executed with the same set of validators (under partial synchrony as well).
The easiest example of the gap between them arises in the case of an interchain protocol with two constituent blockchains $\PI_0$ and $\PI_1$, each running Tendermint with $n=3f+1$ validators and a quorum of $2f+1$.
Whereas the only non-trivial pareto-optimal interchain protocol instantiated with these blockchains (and achievable by the timestamping protocol) satisfies $f$-liveness and $2f+2$-slashable-safety, the (pareto-optimal) HotStuff protocol executed by these $2n$ validators with a quorum of $4f+2$ satisfies $2f$-liveness and $2f+2$-slashable-safety, and the closure of its quorum and fail-prone systems subsume those of the interchain protocol.

\section*{Acknowledgements}
We thank Yifei Wang and Dionysis Zindros for several insightful discussions on this work. 
Ertem Nusret Tas is supported by the Stanford Center for Blockchain Research. 

\bibliographystyle{plain}
\bibliography{references}

\begin{thebibliography}{10}

\bibitem{komodo}
{Komodo. Advanced blockchain technology, focused on freedom}.
\newblock \url{https://docs.komodoplatform.com/whitepaper/introduction.html}.

\bibitem{recursive-tendermint}
Ics ?: Recursive tendermint.
\newblock Website, 2019.
\newblock \url{https://github.com/cosmos/ibc/issues/547}.

\bibitem{coinmarketcap}
Today's cryptocurrency prices by market cap.
\newblock Website, 2021.
\newblock \url{https://coinmarketcap.com/}.

\bibitem{axelar}
{Axelar}.
\newblock Website, 2023.
\newblock \url{https://axelar.network/}.

\bibitem{composable}
{Composable finance}.
\newblock Website, 2023.
\newblock \url{https://www.composable.finance/}.

\bibitem{cosmos-sdk-source}
{cosmos/cosmos-sdk: A Framework for Building High Value Public Blockchains}.
\newblock Website, 2023.
\newblock \url{https://github.com/cosmos/cosmos-sdk}.

\bibitem{ibc-relayer}
{cosmos/relayer: An IBC relayer for IBC-Go}.
\newblock Website, 2023.
\newblock \url{https://github.com/cosmos/relayer}.

\bibitem{electron}
{Electron Labs}.
\newblock Website, 2023.
\newblock \url{https://www.electronlabs.org/}.

\bibitem{mintscan-example-tx-client-update}
{Example transaction for updating an IBC client}.
\newblock Website, 2023.
\newblock
  \url{https://www.mintscan.io/cosmos/txs/271e872f0ca7b247d0685b51a390b891a8d42eb179db2dc3e8bdcc7836cea850}.

\bibitem{client-update-frequency}
{Frequency of light client updates v.s. IBC packets}.
\newblock Website, 2023.
\newblock
  \url{https://github.com/cosmos/relayer/blob/main/docs/advanced_usage.md#auto-update-light-client}.

\bibitem{graphql}
{GraphQL | A query language for your API}.
\newblock Website, 2023.
\newblock \url{https://graphql.org/}.

\bibitem{ibc-token-transfer}
{IBC Token Transfer | Cosmos Developer Portal}.
\newblock Website, 2023.
\newblock
  \url{https://tutorials.cosmos.network/academy/3-ibc/5-token-transfer.html}.

\bibitem{map-of-zones}
{Map of zones - Cosmos network explorer}.
\newblock Website, 2023.
\newblock \url{https://mapofzones.com/}.

\bibitem{sunny-mesh}
Mesh security.
\newblock Youtube, 2023.
\newblock \url{https://www.youtube.com/watch?v=GjX4ejD_cRA&t=4670s}.

\bibitem{mesh-sec-github}
Mesh security.
\newblock Website, 2023.
\newblock \url{https://github.com/osmosis-labs/mesh-security}.

\bibitem{mintscan-cosmos-blocks}
{Mintscan - Cosmos Blocks}.
\newblock Website, 2023.
\newblock \url{https://www.mintscan.io/cosmos/blocks}.

\bibitem{networkx}
{NetworkX: Network Anlaysis in Python}.
\newblock Website, 2023.
\newblock \url{https://networkx.org/}.

\bibitem{polymer}
{Polymer Labs}.
\newblock Website, 2023.
\newblock \url{https://www.polymerlabs.org/}.

\bibitem{cross-chain-validation-blog}
{Replicated vs. Mesh Security (blog post from Informal Systems)}.
\newblock Website, 2023.
\newblock
  \url{https://informal.systems/2022/11/04/replicated-vs-mesh-security}.

\bibitem{source-code}
{Source code for mesh security analysis}.
\newblock Website, 2023.
\newblock \url{https://github.com/SebastianElvis/mapofzones-crawler}.

\bibitem{msgupdateclient}
{Transport, Authentication, and Ordering Layer - Clients}.
\newblock Website, 2023.
\newblock
  \url{https://tutorials.cosmos.network/academy/3-ibc/4-clients.html#updating-a-client}.

\bibitem{synchotstuff}
Ittai Abraham, Dahlia Malkhi, Kartik Nayak, Ling Ren, and Maofan Yin.
\newblock Sync hotstuff: Simple and practical synchronous state machine
  replication.
\newblock In {\em {IEEE} Symposium on Security and Privacy}, pages 106--118.
  {IEEE}, 2020.

\bibitem{albassam2019lazyledger}
Mustafa Al-Bassam.
\newblock Lazyledger: A distributed data availability ledger with client-side
  smart contracts.
\newblock {\em arXiv:1905.09274}, 2019.

\bibitem{pikachu}
Sarah Azouvi and Marko Vukolic.
\newblock Pikachu: Securing pos blockchains from long-range attacks by
  checkpointing into bitcoin pow using taproot.
\newblock {\em arXiv:2208.05408}, 2022.

\bibitem{tendermint_thesis}
Ethan Buchman.
\newblock Tendermint: Byzantine fault tolerance in the age of blockchains,
  2016.

\bibitem{tendermint}
Ethan Buchman, Jae Kwon, and Zarko Milosevic.
\newblock The latest gossip on {BFT} consensus.
\newblock {\em arXiv:1807.04938}, 2018.

\bibitem{casper}
Vitalik Buterin and Virgil Griffith.
\newblock {Casper} the friendly finality gadget.
\newblock {\em arXiv:1710.09437}, 2019.

\bibitem{CachinT19}
Christian Cachin and Bj{\"{o}}rn Tackmann.
\newblock Asymmetric distributed trust.
\newblock In {\em {OPODIS}}, volume 153 of {\em LIPIcs}, pages 7:1--7:16.
  Schloss Dagstuhl - Leibniz-Zentrum f{\"{u}}r Informatik, 2019.

\bibitem{pbft}
Miguel Castro and Barbara Liskov.
\newblock Practical byzantine fault tolerance.
\newblock In {\em {OSDI}}, pages 173--186. {USENIX} Association, 1999.

\bibitem{streamlet}
Benjamin~Y. Chan and Elaine Shi.
\newblock Streamlet: Textbook streamlined blockchains.
\newblock In {\em {AFT}}, pages 1--11. {ACM}, 2020.

\bibitem{DLS88}
Cynthia Dwork, Nancy~A. Lynch, and Larry~J. Stockmeyer.
\newblock Consensus in the presence of partial synchrony.
\newblock {\em J. {ACM}}, 35(2):288--323, 1988.

\bibitem{ledger-combiners}
Matthias Fitzi, Peter Gazi, Aggelos Kiayias, and Alexander Russell.
\newblock Ledger combiners for fast settlement.
\newblock In {\em {TCC} {(1)}}, volume 12550 of {\em Lecture Notes in Computer
  Science}, pages 322--352. Springer, 2020.

\bibitem{bitcoin-timestamp}
Thomas Hepp, Patrick Wortner, Alexander Sch{\"{o}}nhals, and Bela Gipp.
\newblock Securing physical assets on the blockchain: Linking a novel object
  identification concept with distributed ledgers.
\newblock In {\em CRYBLOCK@MobiSys}, pages 60--65. {ACM}, 2018.

\bibitem{lamport82}
Leslie Lamport, Robert~E. Shostak, and Marshall~C. Pease.
\newblock The byzantine generals problem.
\newblock {\em {ACM} Trans. Program. Lang. Syst.}, 4(3):382--401, 1982.

\bibitem{lewispyeroughgardenccs}
Andrew Lewis-Pye and Tim Roughgarden.
\newblock How does blockchain security dictate blockchain implementation?
\newblock In {\em Conference on Computer and Communications Security}, CCS '21.
  {ACM}, 2021.

\bibitem{malkhi98}
Dahlia Malkhi and Michael~K. Reiter.
\newblock Byzantine quorum systems.
\newblock {\em Distributed Comput.}, 11(4):203--213, 1998.

\bibitem{bitcoin}
Satoshi Nakamoto.
\newblock Bitcoin: A peer-to-peer electronic cash system.
\newblock \url{https://bitcoin.org/bitcoin.pdf}, 2008.

\bibitem{snapandchat}
Joachim Neu, Ertem~Nusret Tas, and David Tse.
\newblock {Snap-and-Chat} protocols: System aspects.
\newblock {\em arXiv:2010.10447}, 2020.

\bibitem{ebbandflow}
Joachim Neu, Ertem~Nusret Tas, and David Tse.
\newblock Ebb-and-flow protocols: {A} resolution of the availability-finality
  dilemma.
\newblock In {\em {IEEE} Symposium on Security and Privacy}, pages 446--465.
  {IEEE}, 2021.

\bibitem{acc_gadget}
Joachim Neu, Ertem~Nusret Tas, and David Tse.
\newblock The availability-accountability dilemma and its resolution via
  accountability gadgets.
\newblock In {\em Financial Cryptography}, volume 13411 of {\em Lecture Notes
  in Computer Science}, pages 541--559. Springer, 2022.

\bibitem{sleepy}
Rafael Pass and Elaine Shi.
\newblock The sleepy model of consensus.
\newblock In {\em {ASIACRYPT} 2017}, pages 380--409. Springer, 2017.

\bibitem{veriblock-whitepaper}
Maxwell Sanchez and Justin Fisher.
\newblock Proof-of-proof: A decentralized, trustless, transparent, and scalable
  means of inheriting proof-of-work security.
\newblock Website, 2018.
\newblock
  \url{https://veriblock.org/wp-content/uploads/2018/03/PoP-White-Paper.pdf}.

\bibitem{sankagiri_clc}
Suryanarayana Sankagiri, Xuechao Wang, Sreeram Kannan, and Pramod Viswanath.
\newblock Blockchain {CAP} theorem allows user-dependent adaptivity and
  finality.
\newblock In {\em Financial Cryptography {(2)}}, volume 12675 of {\em Lecture
  Notes in Computer Science}, pages 84--103. Springer, 2021.

\bibitem{forensics}
Peiyao Sheng, Gerui Wang, Kartik Nayak, Sreeram Kannan, and Pramod Viswanath.
\newblock {BFT} protocol forensics.
\newblock In {\em {CCS}}, pages 1722--1743. {ACM}, 2021.

\bibitem{btc-pos}
Ertem~Nusret Tas, David Tse, Fisher Yu, Sreeram Kannan, and Mohammad~Ali
  Maddah{-}Ali.
\newblock Bitcoin-enhanced proof-of-stake security: Possibilities and
  impossibilities.
\newblock {\em arXiv:2207.08392}, 2022.
\newblock To appear in IEEE S\&P 2023.

\bibitem{trustboost}
Xuechao Wang, Peiyao Sheng, Sreeram Kannan, Kartik Nayak, and Pramod Viswanath.
\newblock Trustboost: Boosting trust among interoperable blockchains.
\newblock {\em {IACR} Cryptol. ePrint Arch.}, page 1428, 2022.

\bibitem{yin2018hotstuff}
Maofan Yin, Dahlia Malkhi, Michael~K. Reiter, Guy Golan{-}Gueta, and Ittai
  Abraham.
\newblock Hotstuff: {BFT} consensus with linearity and responsiveness.
\newblock In {\em {PODC}}, pages 347--356. {ACM}, 2019.

\bibitem{ccc}
Alexei Zamyatin, Mustafa Al{-}Bassam, Dionysis Zindros, Eleftherios
  Kokoris{-}Kogias, Pedro Moreno{-}Sanchez, Aggelos Kiayias, and William~J.
  Knottenbelt.
\newblock Sok: Communication across distributed ledgers.
\newblock In {\em Financial Cryptography {(2)}}, volume 12675 of {\em Lecture
  Notes in Computer Science}, pages 3--36. Springer, 2021.

\end{thebibliography}

\appendix

\section{Cross-staking Under Synchrony}
\label{sec:appendix-smr}

Limits of quorum and fail-prone systems achievable by any SMR protocol under synchrony is given below:
\begin{theorem}[Safety-Liveness Trade-off under Synchrony]
\label{thm:sync-safety-liveness-converse}
For every SMR protocol $\PI$ that is $\Q$-live and $\B_s$-safe under synchrony, it holds that $\forall Q \in \Q$ and $B_s \in \B_s \colon Q \not\subseteq B_s$.
\end{theorem}
The corresponding achievability result is stated by the following theorem:
\begin{theorem}
\label{thm:sync-hotstuff}
Any tuple of quorum and fail-prone systems $(\Q$, $\B_s$, $\B_a)$ achievable by a pareto-optimal SMR protocol under synchrony can be achieved by Sync HotStuff~\cite{synchotstuff} executed with the quorums in $\Q$.
\end{theorem}
Theorem~\ref{thm:sync-hotstuff} 
gives a construction achieving all pareto-optimal SMR protocols under synchrony.
Its proof is the \emph{same} as the security proof for Sync HotStuff~\cite{synchotstuff}, except that the quorums are changed to be the sets in $\Q$.

\section{Proofs}
\label{sec:appendix-proofs}
\subsection{Proof of Theorem~\ref{thm:interchain-timestamping-security}}
\label{sec:appendix-timestamping-security}
\begin{proof}[Proof of Theorem~\ref{thm:interchain-timestamping-security}]

We first prove the theorem for $k=2$.

\textbf{Safety:}
Suppose the consumer blockchain is safe (w.o.p.) for all PPT $\mathcal{A}$.
Then, $\chain^{C,\client_1}_{r_1} \preceq \chain^{C,\client_2}_{r_2}$ or vice versa (\ie, consumer chains are prefixes) for any two clients $\client_1$ and $\client_2$ and rounds $r_1$ and $r_2$, implying that the union $\mathcal{T}_C$ of all finalized consumer blocks observed by the clients across all rounds is a chain.
In this case, for any two valid timestamps $\ckpt_i$ and $\ckpt_j$ with the consumer blocks $B_i$ and $B_j$ at the preimages of their hashes, it holds that either $B_i \preceq B_j$ or vice versa.
Hence, for any client $\client$ and round $r$, the timestamped ledger $\aux^{\client}_r$ is a chain within $\mathcal{T}_C$.
This implies that $\aux^{\client_1}_{r_1} \preceq \aux^{\client_2}_{r_2}$ or vice versa for any two clients $\client_1$ and $\client_2$ and rounds $r_1$ and $r_2$.

Next, suppose the provider blockchain is safe (w.o.p.) for all PPT $\mathcal{A}$.
Then, without loss of generality, $\chain^{P,\client_1}_{r_1} \preceq \chain^{P,\client_2}_{r_2}$ (\ie, provider header chains are prefixes) for any two clients $\client_1$ and $\client_2$ and rounds $r_1$ and $r_2$.
Let $\ckpt_i$, $i \in [m_1]$, and $\ckpt_j$, $j \in [m_2]$, $m_1 \leq m_2$, denote the sequence of valid timestamps in $\client_1$'s and $\client_2$'s views at rounds $r_1$ and $r_2$ respectively.
Note that the sequence observed by $\client_1$ is a prefix of $\client_2$'s sequence.

Starting from the genesis consumer block, let $B_1$ denote the first consumer block in $\aux^{\client_1}_{r_1}$ that is not available or finalized in $\client_2$'s view at round $r_2$, and define $i_1$ as the index of the first valid timestamp whose preimage block has $B_1$ in its prefix
(if there is no such block $B_1$, let $i_1 = \infty$).
Similarly, let $B_2$ denote the first consumer block in $\aux^{\client_2}_{r_2}$ that is not available or finalized in $\client_1$'s view at round $r_1$, and define $i_2$ as the index of the first valid timestamp whose preimage block has $B_2$ in its prefix
(if there is no such block $B_2$, let $i_2 = \infty$).

By the collision-resistance of the hash function and the security of the pre-commit signatures, for any valid timestamp $\ckpt_i$, $i \in [m_1]$, with index $i<\min(i_1,i_2)$, the condition at Line~\ref{line:btc2} of Alg.~\ref{alg.timestamping} is false for $\client_1$ if and only if it is false for $\client_2$. 
Similarly, the clients must have obtained the same timestamped ledger $\aux$ at Line~\ref{line:update} before the stalling condition is triggered at Line~\ref{line:btc2} for the client that stalls at the earlier timestamp.
Moreover $\aux \preceq \clean(\aux,\chain)$ for all ledgers $\aux$ and chains $\chain$.
Thus, if $i_1 = \infty$,  $\aux^{\client_1}_{r_1} \preceq \aux^{\client_2}_{r_2}$.
If $i_1<\infty$, then $i_2 = \infty$, due to Line~\ref{line:btc2} being triggered earlier for $\client_2$.
Thus, if $i_1<i_2$, then $\aux^{\client_2}_{r_2} \prec \aux^{\client_1}_{r_1}$, and if $i_2 \leq i_1$, then $\aux^{\client_1}_{r_1} \preceq \aux^{\client_2}_{r_2}$, concluding the safety proof.

If both constituent blockchains are not safe with non-negligible probability for some PPT $\mathcal{A}$, then two clients observing conflicting consumer and provider chains output conflicting timestamped ledgers with non-negligible probability.

\textbf{Slashable Safety:} As shown by the proof above, the timestamping protocol satisfies safety (w.o.p.) for a given set of adversarial validators iff at least \emph{one} of the constituent blockchains is safe (w.o.p.) for all PPT $\mathcal{A}$. 
Thus, if safety is violated, it must be the case that safety is violated in all of the constituent blockchains. 

\textbf{Liveness:}
Suppose both constituent blockchains are live (w.o.p.) for all PPT $\mathcal{A}$.
Then, any transaction $\tx$ input to $\PI_0$ at round $r$ appears in the finalized consumer chain in the view of all online clients, including honest validators, at all rounds $r' \geq \max(\GST,r)+T_C$.

By round $\max(\GST,r)+T_C$, each honest validator sends a valid timestamp for the finalized consumer block containing $\tx$ in its prefix.
All of these timestamps appear in the provider chains in the view of all online clients at all rounds $r' \geq \max(\GST,r)+T_C+T_P$.
Moreover, for any chain $\chain$ and ledger $\aux$, if $\tx \in \chain$, then $\tx \in \clean(\aux,\chain)$.  
Hence, for any online client $\client$ and round $r' \geq \max(\GST,r)+T_C+T_P$, $\tx \in \aux^{\client}_{\max(\GST,r)+T_C+T_P}$, concluding the liveness proof.

If one of the constituent blockchains is not live with non-negligible probability for some PPT $\mathcal{A}$, then $\mathcal{A}$ can ensure that a client always outputs an empty chain for the blockchain that is not live, thus violate the liveness of the timestamping protocol.

\textbf{Generalizing to $\mathbf{k+1}$ blockchains:}
We iteratively apply the proof for the case of $2$ blockchains to $k+1$ blockchains.
If any of the constituent blockchains (\eg, $\PI_{i^*}$) is safe (w.o.p.) for all PPT $\mathcal{A}$, then the timestamped ledger of $\PI_{i^*}$ is safe by the proof above.
By the same proof, the timestamped ledger of $\PI_{i^*-1}$ is also safe, and so on, implying the safety of the timestamped ledger of $\PI_0$, \ie the $\PI_I$ chain.
In contrast, if all blockchains are live (w.o.p.) for all PPT $\mathcal{A}$, then by the liveness proof above, the timestamped ledger of $\PI_{k-1}$ is live, and so on, implying the liveness of the timestamped ledger of $\PI_0$, \ie, the $\PI_I$ chain.
\end{proof}
\subsection{Proofs for Section~\ref{sec:interchain-protocols}}
\label{sec:appendix-interchain-proofs}
\begin{proof}[Proof of Theorem~\ref{thm:safety-liveness-interchain-converse}]
Towards contradiction, suppose there exists an interchain protocol executed using the blockchains $\PI_i$, $i \in [k]$, with validator sets $\UU^i$, quorum systems $\Q^i$ and fail-prone systems $\B^i_s$ such that $\exists Q^1, Q^2 \in \Q^I$ and $B_s \in \B^I_s \colon f_Q(Q^1) \cap f_Q(Q^2) \subseteq f_s(B_s)$.
We will consider the following worlds, and through an indistinguishability argument, show the existence of a world with a safety violation.

Let $U^1$ and $U^2$ denote the validators belonging to the chains in $f_Q(Q^1) \backslash f_Q(Q^2)$ and $f_Q(Q^2) \backslash f_Q(Q^1)$ respectively.
Let $U_3$ denote the validators belonging to the chains in $f_Q(Q^1) \cap f_Q(Q^2)$.

\paragraph{World 1:}
\paragraph{Setup.}
All messages sent by the honest validators are delivered to their recipients in the next round. 
There is a single client $\client_1$.
Validators in $Q^1$ are honest and the rest are adversarial.
The environment inputs a single transaction $\tx_1$ to the validators in $Q^1$ at round $0$.
The adversarial validators do not communicate with those in $Q^1$, and do not respond to $\client_1$.
They also ensure that for all $i \in [k]$, $i \notin f_Q(Q^1)$, all clients of $\PI_i$, including $\client_1$ and the validators in $U^1$, output empty ledgers for $\PI_i$ at all times (w.o.p.).
\paragraph{Output.}
By liveness, $\client_1$ outputs the ledger $\langle \tx_1 \rangle$ by round $\Tconfirm$ (w.o.p.).

\paragraph{World 2:}
\paragraph{Setup.}
All messages sent by the honest validators are delivered to their recipients in the next round. 
There is a single client $\client_2$.
Validators in $Q^2$ are honest and the rest are adversarial.
The environment inputs a single transaction $\tx_2$ to the validators in $Q^2$ at round $0$.
The adversarial validators do not communicate with those in $Q^2$, and do not respond to $\client_2$.
They also ensure that for all $i \in [k]$, $i \notin f_Q(Q^2)$, all clients of $\PI_i$, including $\client_2$ and the validators in $U^2$, output empty ledgers for $\PI_i$ at all times (w.o.p.).
\paragraph{Output.}
By liveness, $\client_2$ outputs the ledger $\langle \tx_2 \rangle$ by round $\Tconfirm$ (w.o.p.).

\paragraph{World 3:}
\paragraph{Setup.}
World 3 is a hybrid world. 
There are two client $\client_1$ and $\client_2$.
At round $0$, the environment inputs the transaction $\tx_1$ to the validators in $Q^1$ and $\tx_2$ to the validators in $Q^2$.
Validators in $Q^1 \backslash (U_3 \cap B_s)$ and $Q^2 \backslash (U_3 \cap B_s)$ are honest, those in $U_3 \cap B_s$ are adversarial, and the rest have crashed.
All messages from the validators in $U^1$ to $\client_2$ and $U^2$, and from those in $U^2$ to $\client_1$ and $U^1$ are delayed by the adversary until after round $\Tconfirm$.
Similarly, all messages from the validators in $Q^1 \backslash Q^2$ to $\client_2$ and $Q^2 \backslash Q^1$, and from those in $Q^2 \backslash Q^1$ to $\client_1$ and $Q^1 \backslash Q^2$ are delayed by the adversary until after round $\Tconfirm$.

The adversary ensures that for all $i \in [k]$, $i \notin f_Q(Q^1) \cup f_Q(Q^2)$, all clients of $\PI_i$, including $\client_1$ and $\client_2$ and the validators in $U^1 \cup U^2$, output empty ledgers for $\PI_i$ at all times (w.o.p.), since the validators belonging to $Q^1 \backslash Q^2$ and those belonging to $Q^2 \backslash Q^1$ on these blockchains are isolated from each other, and do not constitute a quorum by themselves.
Similarly, $\client_1$ and the validators in $U^1$ output empty ledgers for $\PI_i$, $i \in f_Q(Q^2) \backslash f_Q(Q^1)$; whereas $\client_2$ and the validators in $U^2$ output empty ledgers for $\PI_i$, $i \in f_Q(Q^1) \backslash f_Q(Q^2)$, until round $\Tconfirm$.
As $f_Q(Q^1) \cap f_Q(Q^2) \subseteq f_s(B_s)$, the adversarial validators in $U_3 \cap B_s$ emulate a split-brain attack via a safety violation on the protocols $\PI_i$, $i \in f_Q(Q^1) \cap f_Q(Q^2)$: 
One brain simulates the execution in world 1 towards $\client_1$ and the validators in $U^1 \cap Q^1$ with transaction $\tx_1$, and the other simulates the execution in world 2 towards $\client_2$ and the validators in $U^2 \cap Q^2$ with transaction $\tx_2$.

\paragraph{Output.}
The worlds 1 and 3 are indistinguishable in $\client_1$'s view and the views of the honest validators in $U^1 \cap Q^1$ (w.o.p.).
Hence, $\client_1$ outputs the ledger $\langle \tx_1 \rangle$ by round $\Tconfirm$.
Similarly, the worlds 2 and 3 are indistinguishable in $\client_2$'s view and the views of the honest validators in $U^2 \cap Q^2$ (w.o.p.).
Hence, $\client_2$ outputs the ledger $\langle \tx_2 \rangle$ by round $\Tconfirm$. 
However, this implies a safety violation when only the validators in $U_3 \cap B_s$ are adversarial, which is a contradiction.
\end{proof}
\begin{proof}[Proof of Theorem~\ref{thm:acc-safety-liveness-interchain-converse}]
Proof follows the outline of the proof of \cite[Theorem B.1]{forensics}.
Towards contradiction, suppose there exists an interchain protocol executed using the blockchains $\PI_i$, $i \in [k]$, with validator sets $\UU^i$, quorum systems $\Q^i$ and fail-prone systems $\B^i_a$ such that $\exists Q^1, Q^2 \in \Q^I$ and $B_a \in \B^I_a \colon f_Q(Q^1) \cap f_Q(Q^2) \subset f_a(B_a)$.
We will consider the following worlds, and through an indistinguishability argument, show the existence of a world where a client identifies an honest validator as a protocol violator.
We assume a synchronous network throughout the following worlds.

Let $U^1$ and $U^2$ denote the validators belonging to the chains in $f_Q(Q^1) \backslash f_Q(Q^2)$ and $f_Q(Q^2) \backslash f_Q(Q^1)$ respectively.
Let $U_3$ denote the validators belonging to the chains in $f_Q(Q^1) \cap f_Q(Q^2)$.
Suppose $f_Q(Q^1) \cup f_Q(Q^2) = [k]$ (if not, we can redefine $Q^2$ to be a large enough quorum to cover the blockchains not in the union).

\paragraph{World 1:}
\paragraph{Setup.}
There is a single client $\client_1$.
Validators in $U^1 \cup Q^1$ are honest and the rest are adversarial.
The environment inputs a single transaction $\tx_1$ to the validators in $U^1$ at round $0$.
The adversarial validators do not communicate with those in $U^1 \cup Q^1$, and do not respond to $\client_1$.
They also ensure that for all $i \in [k]$, $i \notin f_Q(Q^1)$, all clients of $\PI_i$, including $\client_1$ and the validators in $U^1$, output empty ledgers for $\PI_i$ at all times (w.o.p.).
\paragraph{Output.}
By liveness, $\client_1$ outputs the ledger $\langle \tx_1 \rangle$ by round $\Tconfirm$ (w.o.p.).

\paragraph{World 2:}
There is a single client $\client_2$.
Validators in $U^2 \cup Q^2$ are honest and the rest are adversarial.
The environment inputs a single transaction $\tx_2$ to the validators in $U^2$ at round $0$.
The adversarial validators do not communicate with those in $U^2 \cup Q^2$, and do not respond to $\client_2$.
They also ensure that for all $i \in [k]$, $i \notin f_Q(Q^2)$, all clients of $\PI_i$, including $\client_2$ and the validators in $U^2$, output empty ledgers for $\PI_i$ at all times (w.o.p.).
\paragraph{Output.}
By liveness, $\client_2$ outputs the ledger $\langle \tx_2 \rangle$ by round $\Tconfirm$ (w.o.p.).

\paragraph{World 3:}
\paragraph{Setup.}
There is a single client $\client_1$.
Validators in $U^1$ are honest and the rest are adversarial.
At round $0$, the environment inputs the transaction $\tx_1$ to the validators in $U^1$ and $\tx_2$ to the validators in $U^2$.

The adversarial validators in $U^2$ do not communicate with $\client_1$ and the honest validators in $U^1$, and ignore all messages sent by them.
The adversary ensures that $\client_1$ and the honest validators in $U^1$ output empty ledgers for $\PI_i$, $i \notin f_Q(Q^1)$.
The adversarial validators in $U_3$ emulate a split-brain attack via a safety violation in the blockchains $\PI_i$, $i \in f_Q(Q^1) \cap f_Q(Q^2)$: 
One brain simulates the execution in world 1 towards $\client_1$ and the honest validators in $U^1$ with transaction $\tx_1$, and the other simulates the execution in world 2 with the adversarial validators in $U^2$ and transaction $\tx_2$.

\paragraph{Output.}
The worlds 1 and 3 are indistinguishable in $\client_1$'s view and the views of the honest validators in $U^1$ (w.o.p.).
Hence, $\client_1$ outputs the ledger $\langle \tx_1 \rangle$ by round $\Tconfirm$.
Since the adversarial validators in $U^2$ simulate the execution in world 2, by emulating the client of world 2, an adversarial validator outputs the ledger $\langle \tx_2 \rangle$ by round $\Tconfirm$.
Since this is a safety violation, a new client $c_3$ asks the validators for their transcripts, upon which the adversarial validators in $U^2$ reply with transcripts that omit the messages received from those in $U^1$.
The client $\client_3$ then runs the forensic protocol with these transcripts. 
Since $f_Q(Q^1) \cap f(Q^2) \subset f_a(B_a)$, $\client_3$ outputs a proof that identifies at least one adversarial validator in $U^2$.

\paragraph{World 4:}
\paragraph{Setup.}
There is a single client $\client_2$.
Validators in $U^2$ are honest and the rest are adversarial.
At round $0$, the environment inputs the transaction $\tx_2$ to the validators in $U^1$ and $\tx_2$ to the validators in $U^2$.

The adversarial validators in $U^1$ do not communicate with $\client_2$ and the honest validators in $U^2$, and ignore all messages sent by them.
The adversary ensures that $\client_2$ and the honest validators in $U^2$ output empty ledgers for $\PI_i$, $i \notin f_Q(Q^2)$.
The adversarial validators in $U_3$ emulate a split-brain attack via a safety violation in the blockchains $\PI_i$, $i \in f_Q(Q^1) \cap f_Q(Q^2)$: 
One brain simulates the execution in world 2 towards $\client_2$ and the honest validators in $U^2$ with transaction $\tx_2$, and the other simulates the execution in world 1 with the adversarial validators in $U^1$ and transaction $\tx_1$.

\paragraph{Output.}
As $f_Q(Q^1) \cap f_Q(Q^2) \subseteq f_s(B_s)$, the worlds 2 and 4 are indistinguishable in $\client_2$'s view and the views of the honest validators in $U^2$ (w.o.p.).
Hence, $\client_2$ outputs the ledger $\langle \tx_2 \rangle$ by round $\Tconfirm$.
Since the validators in $U^1$ simulate the execution in world 1, by emulating the client of world 1, an adversarial validator outputs the ledger $\langle \tx_1 \rangle$ by round $\Tconfirm$.
Since this is a safety violation, a new client $c_3$ asks the validators for their transcripts, upon which the adversarial validators in $U^1$ reply with transcripts that omit the messages received from those in $U^2$.
Note that the worlds 3 and 4 (and the transcripts received therein) are indistinguishable in the view of $\client_3$ (w.o.p.).
Thus, upon running the forensic protocol with these transcripts, $\client_3$ outputs a proof that identifies at least one validator in $U^2$.
However, the validators in $U^2$ are honest in world 4, which is a contradiction.

\begin{center}
    ***
\end{center}

Finally, we prove the second part of the theorem by contradiction. 
Suppose $\exists B_a \in \B_a, j \in [k] \colon \exists B^j_a \in \B^j_a, B_a \cap \UU^j \supset B^j_a$.
Consider the world, where the validators in the set $B_a$ are adversarial and cause a safety violation in $\PI_I$.
Since the validators from $\PI_j$ are irrefutably identified by a forensic protocol after this safety violation, there must also have been a safety violation in $\PI_j$.

Now, consider a world where the validators in $B'_a = B^j_a \cup (B_a \backslash \UU^j)$ are adversarial.
In this new world, the validators in $(B_a \backslash \UU^j)$ emulate their execution from the previous world.
The validators in $B^j_a$ again cause a safety violation in $\PI_j$ in the same way as in the previous world.
Hence, there is again a safety violation in $\PI_I$.
However, the set $B'_a$ of adversarial validators is a subset of $B_a$, implying that the set of adversarial validators identified by the forensic protocol has to be a strict subset of $B_a$.
Thus, $B'_a \in \B_a$ and $B'_a \subset B_a$, which is a contradiction.
\end{proof}
\begin{proof}[Proof of Theorem~\ref{thm:data-limitation}]
Towards contradiction, suppose $\exists Q \in \Q^I \colon \forall j \in [k'], i_j \notin f_Q(Q)$.
Consider the execution, where all validators in $Q$ belonging to the chains in $f_Q(Q)$ are honest and the rest have crashed.
Note that these validators can still finalize $\PI_I$ blocks since the adversary can ensure that all chains not in $f_Q(Q)$ output empty ledgers in the views of all clients even when all validators in $Q$ are honest.
Since these validators do not check for the availability of the proposed blocks, an adversarially proposed valid, yet unavailable block can be output as part of the chain of an honest validator.
However, in this case, the external clients checking for the data and receiving this chain lose liveness, implying that the protocol $\PI_I$ cannot satisfy liveness even when all validators in $Q$ belonging to the chains in $f_Q(Q)$ are honest.
This is a contradiction.
\end{proof}
\begin{proof}[Proof of Theorem~\ref{thm:optimality-of-trustboost}]
Consider an upper-boundary point $(\Q,\B_s,\B_a)$ for the $k$ blockchains $\PI_i$, $i \in [k]$ and a Trustboost protocol \cite{trustboost} based on HotStuff \cite{yin2018hotstuff}, with quorums determined by $\DQ = F_Q(\Q)$ and consisting of these $k$ blockchains.
Let $(\DQ,\D_s,\D_a)$ denote the upper-boundary property point determined by $\DQ$.

In Trustboost, each blockchain $\PI_i$ simulates a validator of HotStuff via a smart contract, and the validators simulated by the chains exchange messages via the CCC abstraction.
In this paradigm, a blockchain with safety corresponds to a validator that never commits a violation of the HotStuff protocol rules such that their violation would provably identify the validator as a protocol violator (\cf \cite{yin2018hotstuff,forensics} for a list of these rules).
An example of such violations is signing and broadcasting prepare votes for two conflicting proposal within the same view. 
Then, the HotStuff instance run by these validators provides $\D_a$-slashable safety by a generalization of HotStuff's accountability proof in \cite{snapandchat,forensics}.
As accountable safety implies safety, HotStuff also satisfies safety (w.o.p., for all PPT $\mathcal{A}$) when all chains in a set within $\D_s$ are safe (w.o.p., for all PPT $\mathcal{A}$).

On the other hand, a blockchain with liveness corresponds to a validator that emulates an honest validator in at least one execution trace of HotStuff (note that this validator might participate in conflicting execution traces simultaneously, \eg, by double-signing blocks).
Thus, when the chains in a set within $\DQ$ are all live, there exists an execution trace of HotStuff such that an implementation with only that execution trace would constitute a secure protocol.
Thus, HotStuff satisfies liveness (w.o.p., for all PPT $\mathcal{A}$) when the chains in a set within $\DQ$ are all live (w.o.p., for all PPT $\mathcal{A}$).

Since the tuple of property systems $(\DQ,\D_s,\D_a)$ satisfied by the Trustboost protocol above is an upper-boundary property point, by Theorem~\ref{thm:interchain-quorum-optimality-2}, its quorum and fail-prone systems is the upper-boundary point $(\Q, \B_s, \B_a)$ such that $\DQ = f_Q(\Q)$, $\D_s = f_s(\B_s)$ and $\D_a = f_a(\B_a)$.
Hence, for any positive integer $k$, all upper-boundary points for interchain protocols with $k$ blockchains can be achieved by a Trustboost protocol instantiated with HotStuff and the same blockchains.
Then, again by Theorem~\ref{thm:interchain-quorum-optimality-2}, quorum and fail-prone systems of all pareto-optimal interchain protocols are upper-boundary points and can be achieved by a Trustboost protocol instantiated with HotStuff and the same blockchains.
\end{proof}
\subsection{Proof of Theorem~\ref{thm:interchain-quorum-optimality-2}}
\label{sec:appendix-theorem-opt-2}
\begin{lemma}
\label{lem:interchain-quorum-pareto-optimality-1}
Consider an upper-boundary point with quorum system $\Q^I$ for interchain protocols executed using the blockchains $\PI_i$, $i \in [k]$, with validator sets $\UU^i$ and quorum systems $\Q^i$.
Then, for all $Q \in \Q^I, i \in [k]$, it holds that either $Q \cap \UU^i \in \Q^i$ or $Q \cap \UU^i = \emptyset$.
Moreover, if $\exists Q \in \Q^I, j \in [k] \colon Q \cap \UU^j \in \Q^j$, then $\forall Q' \in \Q^j \colon Q' \cup (Q \backslash \UU^j) \in \Q^I$.
\end{lemma}
\begin{proof}[Proof of Lemma~\ref{lem:interchain-quorum-pareto-optimality-1}]
Towards contradiction, suppose $\exists Q \in \Q^I, j \in [k] \colon Q \cap \UU^j \neq \emptyset$ and $Q \cap \UU^j \notin \Q^i$ for an upper-boundary point $P = (\Q^I,\B^I_s,\B^I_a)$.

Suppose $\exists Q^j \in \Q^j \colon Q \cap \UU^j \subset Q^j$.
Consider the point $P'$ with the same quorum and fail-prone systems as $P$ except that in place of the set $Q$, it has the set $Q' = Q \backslash \UU^j$ in its quorum system.
Since $P$ satisfies Theorems~\ref{thm:safety-liveness-interchain-converse} and~\ref{thm:acc-safety-liveness-interchain-converse}, and $f_Q(Q) = f_Q(Q')$, $P'$ also satisfies the same theorems.
As $Q' \subset Q$, $P$ cannot be an upper-boundary point.

Next, suppose $\exists Q^j \in \Q^j \colon Q \cap \UU^j \supset Q^j$.
Consider the point $P'$ with the same quorum and fail-prone systems as $P$ except that instead of the set $Q$, it has the set $Q' = (Q \backslash \UU^j) \cup Q^j$ in its quorum system.
Since $P$ satisfies Theorems~\ref{thm:safety-liveness-interchain-converse} and~\ref{thm:acc-safety-liveness-interchain-converse}, and $f_Q(Q) = f_Q(Q')$, $P'$ also satisfies the same theorems.
As $Q' \subset Q$, $P$ cannot be an upper-boundary point.

Finally, suppose $\exists Q \in \Q^I, j \in [k] \colon Q \cap \UU^j \in \Q^j$, yet $\exists Q^j \in \Q^j \colon Q^j \cup (Q \backslash \UU^j) \notin \Q^I$.
Consider the point $P'$ with the same quorum and fail-prone systems as $P$ except that $Q^j \cup (Q \backslash \UU^j) = Q' \in \Q'^I$ for the quorum system $\Q'^I$ of that protocol.
Since $\PI_I$ satisfies Theorems~\ref{thm:safety-liveness-interchain-converse} and~\ref{thm:acc-safety-liveness-interchain-converse}, and $f_Q(Q) = f_Q(Q')$, $\PI'_I$ also satisfies the same theorems.
As $\close(\Q^I) \subset \close(\Q'^I)$, $P$ cannot be an upper-boundary point.
\end{proof}
\begin{lemma}
\label{lem:interchain-quorum-pareto-optimality-safety}
Consider an upper-boundary point with fail-prone system $\B_s^I$ for interchain protocols executed using the blockchains $\PI_i$, $i \in [k]$, with validator sets $\UU^i$ and fail-prone systems $\B_s^i$.
Then, for all $B_s \in \B_s^I, i \in [k]$, it holds that either $B_s \cap \UU^i \in \B_s^i$ or $B_s \cap \UU^i = \UU^i$.
Moreover, if $\exists B_s \in \B_s^I, j \in [k] \colon B_s \cap \UU^j \in \B_s^j$, then $\forall B'_s \in \B_s^j \colon B'_s \cup (B_s \backslash \UU^j) \in \B_s^I$.
\end{lemma}
\begin{proof}[Proof of Lemma~\ref{lem:interchain-quorum-pareto-optimality-safety}]
Towards contradiction, suppose $\exists B_s \in \B_s^I, j \in [k] \colon B_s \cap \UU^j \neq \UU^j$ and $B_s \cap \UU^j \notin \B_s^i$ for an upper-boundary point $P = (\Q^I,\B^I_s,\B^I_a)$.

Suppose $\exists B_s^j \in \B_s^j \colon B_s \cap \UU^j \subset B_s^j$.
Consider the point $P'$ with the same quorum and fail-prone systems as $P$ except that in place of the set $B_s$, it has the set $B'_s = B_s \cup B_s^j$ in its quorum system.
Since $P$ satisfies Theorems~\ref{thm:safety-liveness-interchain-converse} and~\ref{thm:acc-safety-liveness-interchain-converse}, and $f_s(B_s) = f_s(B'_s)$, $P'$ also satisfies the same theorems.
As $B_s \subset B'_s$, $P$ cannot be an upper-boundary point.

Next, suppose $\exists B_s^j \in \B_s^j \colon B_s \cap \UU^j \supset B_s^j$.
Consider the point $P'$ with the same quorum and fail-prone systems as $P$ except that instead of the set $B_s$, it has the set $B'_s = B_s \cup \UU^j$ in its quorum system.
Since $P$ satisfies Theorems~\ref{thm:safety-liveness-interchain-converse} and~\ref{thm:acc-safety-liveness-interchain-converse}, and $f_s(B_s) = f_s(B'_s)$, $P'$ also satisfies the same theorems.
As $B_s \subset B'_s$, $P$ cannot be an upper-boundary point.

Finally, suppose $\exists B_s \in \B_s^I, j \in [k] \colon B_s \cap \UU^j \in \B_s^j$, yet $\exists B_s^j \in \B_s^j \colon B_s^j \cup (B_s \backslash \UU^j) \notin \B_s^I$.
Consider the point $P'$ with the same quorum and fail-prone systems as $P$ except that $B_s^j \cup (B_s \backslash \UU^j) = B'_s \in \B_s^{'I}$ for the fail-prone system $\B_s^{'I}$ of that protocol.
Since $\PI_I$ satisfies Theorems~\ref{thm:safety-liveness-interchain-converse} and~\ref{thm:acc-safety-liveness-interchain-converse}, and $f_s(B_s) = f_s(B'_s)$, $\PI'_I$ also satisfies the same theorems.
As $\close(\B_s^I) \subset \close(\B_s^{'I})$, $P$ cannot be an upper-boundary point.
\end{proof}
\begin{lemma}
\label{lem:interchain-quorum-pareto-optimality-acc-safety}
Consider an upper-boundary point with quorum system $\B_a^I$ for interchain protocols executed using the blockchains $\PI_i$, $i \in [k]$, with validator sets $\UU^i$ and fail-prone systems $\B_a^i$.
Then, for all $B_a \in \B_a^I, i \in [k]$, it holds that either $B_a \cap \UU^i \in \B_a^i$ or $B_a \cap \UU^i = \emptyset$.
Moreover, if $\exists B_a \in \B_a^I, j \in [k] \colon B_a \cap \UU^j \in \B_a^j$, then $\forall B'_a \in \B_a^j \colon B'_a \cup (B_a \backslash \UU^j) \in \B_a^I$.
\end{lemma}
\begin{proof}[Proof of Lemma~\ref{lem:interchain-quorum-pareto-optimality-acc-safety}]
Towards contradiction, suppose $\exists B_a \in \B_a^I, j \in [k] \colon B_a \cap \UU^j \neq \emptyset$ and $B_a \cap \UU^j \notin \B_a^i$ for an upper-boundary point $P = (\Q^I,\B^I_s,\B^I_a)$.

Suppose $\exists B_a^j \in \B_a^j \colon B_a \cap \UU^j \subset B_a^j$.
Consider the point $P'$ with the same quorum and fail-prone systems as $P$ except that in place of the set $B_a$, it has the set $B'_a = B_a \cup B_a^j$ in its quorum system.
Since $P$ satisfies Theorems~\ref{thm:safety-liveness-interchain-converse} and~\ref{thm:acc-safety-liveness-interchain-converse}, and $f_a(B_a) = f_a(B'_a)$, $P'$ also satisfies the same theorems.
As $B_a \subset B'_a$, $P$ cannot be an upper-boundary point.

Next, suppose $\exists B_a^j \in \B_a^j \colon B_a \cap \UU^j \supset B_a^j$.
However, this is a violation of Theorem~\ref{thm:acc-safety-liveness-interchain-converse}.
Hence, $P$ cannot be an upper-boundary point.

Finally, suppose $\exists B_a \in \B_a^I, j \in [k] \colon B_a \cap \UU^j \in \B_a^j$, yet $\exists B_a^j \in \B_a^j \colon B_a^j \cup (B_a \backslash \UU^j) \notin \B_a^I$.
Consider the point $P'$ with the same quorum and fail-prone systems as $P$ except that $B_a^j \cup (B_a \backslash \UU^j) = B'_a \in \B_a^{'I}$ for the quorum system $\B_a^{'I}$ of that protocol.
Since $\PI_I$ satisfies Theorems~\ref{thm:safety-liveness-interchain-converse} and~\ref{thm:acc-safety-liveness-interchain-converse}, $f_a(B_a) = f_a(B'_a)$ and $B'_a \cap \UU^j \in \B_a^j$, $\PI'_I$ also satisfies the same theorems.
Hence, as $\close(\B_a^I) \subset \close(\B_a^{'I})$, $P$ cannot be an upper-boundary point.
\end{proof}
\begin{proof}[Proof of Theorem~\ref{thm:interchain-quorum-optimality-2}]
No tuple of quorum and fail-prone systems $(\Q,\B_s,\B_a)$ achievable by an interchain protocol $\PI_I$ can dominate an upper-boundary point $(\Q',\B'_s,\B'_a)$.
To prove this, suppose $(\Q,\B_s,\B_a)$ dominates $(\Q',\B'_s,\B'_a)$.
Then, either $(\Q,\B_s,\B_a)$ does not satisfy one of the Theorems~\ref{thm:safety-liveness-interchain-converse} and~\ref{thm:data-limitation} (or Theorems~\ref{thm:acc-safety-liveness-interchain-converse} and~\ref{thm:data-limitation}) or it satisfies both theorems.
In the first case, no such tuple of quorum and fail-prone systems $(\Q,\B_s,\B_a)$ can be achievable by an interchain protocol by the same theorems.
In the latter case, $(\Q',\B'_s,\B'_a)$ cannot be an upper-boundary point per Definition~\ref{def:interchain-optimality}, which is again a contradiction.

By the reasoning above, if the tuple $(\Q,\B_s,\B_a)$ of quorum and fail-prone systems of an interchain protocol $\PI_I$ is an upper-boundary point, then no tuple $(\Q',\B'_s,\B'_a)$ 
achievable by an interchain protocol can dominate $(\Q,\B_s,\B_a)$, implying that $\PI_I$ is a pareto-optimal protocol. 

Finally, consider such a protocol whose tuple $(\Q,\B_s,\B_a)$ of quorum and fail-prone systems constitutes an upper-boundary point.
By Lemma~\ref{lem:interchain-quorum-pareto-optimality-1}, for all $Q \in \Q^I, i \in [k]$, it holds that either $Q \cap \UU^i \in \Q^i$ or $Q \cap \UU^i = \emptyset$.
Moreover, if $\exists Q \in \Q^I, j \in [k] \colon Q \cap \UU^j \in \Q^j$, then $\forall Q' \in \Q^j \colon Q' \cup (Q \backslash \UU^j) \in \Q^I$.
Given these observations, satisfiability of liveness for the protocol $\PI_I$ can be expressed as a boolean function of the predicates $\ell_i$, $i \in [k]$, such that $\ell_i$ becomes true iff for a given set of adversarial validators, the protocol $\PI_i$ satisfies liveness for all PPT $\mathcal{A}$ (w.o.p.).
This implies the existence of a property system $\DQ$ such that $\PI_I$ is $\DQ$-live, and vice-versa.
Similarly, by Lemma~\ref{lem:interchain-quorum-pareto-optimality-safety}, satisfiability of safety for $\PI_I$ can be expressed as a boolean function of the predicates $s_i$, $i \in [k]$, such that $s_i$ becomes true iff for a given set of adversarial validators, the protocol $\PI_i$ does not satisfy safety for some PPT $\mathcal{A}$ (with non-negligible probability).
This implies the existence of a property system $\D_s$ such that $\PI_I$ is $\D_s$-safe, and vice-versa.
Finally, by Lemma~\ref{lem:interchain-quorum-pareto-optimality-acc-safety}, there exists a property system $\D_a$ such that $\PI_I$ is $\D_a$-accountably-safe, and vice-versa.

For the first direction of the implication, by Theorems~\ref{thm:safety-liveness-interchain-converse} and~\ref{thm:acc-safety-liveness-interchain-converse}, for systems $\DQ = \{f_Q(Q) \colon Q \in \Q\}$, $\D_s = \{f_s(B_s) \colon B_s \in \B_s\}$ and $\D_a = \{f_a(B_a) \colon B_a \in \B_a\}$ of $\PI_I$, it holds that $\forall DQ^1, DQ^2 \in \DQ, D_s \in \D_s, D_a \in \D_a \colon DQ^1 \cap DQ^2 \not\subseteq D_s$ (condition 1) and $DQ^1 \cap DQ^2 \not\subset D_a$ (condition 2).
Theorem~\ref{thm:data-limitation} implies condition 3.
Condition 4 follows from the fact that no tuple $(\Q,\B_s,\B_a)$ of quorum and fail-prone systems satisfying the Theorems~\ref{thm:safety-liveness-interchain-converse},~\ref{thm:acc-safety-liveness-interchain-converse} and~\ref{thm:data-limitation} can dominate that of $\PI_I$.

For the other direction of the implication, the conditions (a)-(b)-(c) imply that the tuple of quorum and fail-prone systems for $\PI_I$ satisfies Theorems~\ref{thm:safety-liveness-interchain-converse},~\ref{thm:acc-safety-liveness-interchain-converse} and~\ref{thm:data-limitation} respectively.
We have observed that the quorum and fail-prone systems of $\PI_I$ constitute an upper-boundary point iff $\PI_I$ is characterized by a tuple of property systems.
Then, the condition (d) implies that there is no tuple of quorum and fail-prone systems dominating that of $\PI_I$ while satisfying the same theorems above.
This concludes the proof.
\end{proof}
\subsection{Proofs for Section~\ref{sec:cross-chain-validation}}
\label{sec:appendix-general-proofs}
\begin{proof}[Proof of Theorem~\ref{thm:sync-safety-liveness-converse}]
Towards contradiction, suppose there exists a protocol with validator set $\UU$, quorum system $\Q$ and fail-prone system $\B_s$ such that $\exists Q \in \Q$ and $B_s \in \B_s \colon Q \subseteq B_s$.
We will show a safety violation when the validators in $B_s$ are adversarial through the following worlds.
Suppose there is a quorum $Q' \in \Q$ such that $Q' \neq Q$
(If there is no such quorum $Q'$, let $Q' = Q$).
Then, define $P_1 = Q \backslash Q'$, $P_2 = \UU \backslash Q$ and $R = Q \cap Q'$.

\paragraph{World 1:}
\paragraph{Setup.}
There is a single client $\client_1$.
Validators in $P_1 \cup R$ are honest and those in $P_2$ are adversarial.
The environment inputs a single transaction $\tx_1$ to the validators in $P_1 \cup R$ at round $0$.
The adversarial validators do not communicate with those in $P_1 \cup R$, and emulate the behavior of the honest validators in $P_2$ in world 2 towards $\client_1$ until round $\Tconfirm$.
\paragraph{Output.}
By liveness, $\client_1$ outputs the ledger $\langle \tx_1 \rangle$ by round $\Tconfirm$ (w.o.p.).

\paragraph{World 2:}
\paragraph{Setup.}
There is a single client $\client_2$.
Validators in $P_2 \cup R$ are honest and those in $P_1$ are adversarial.
The environment inputs a single transaction $\tx_2$ to the validators in $P_2 \cup R$ at round $\Tconfirm+1$.
The adversarial validators do not communicate with those in $P_2 \cup R$, and do not respond to $\client_2$.
\paragraph{Output.}
By liveness, $\client_2$ outputs the ledger $\langle \tx_2 \rangle$ at some round in $(\Tconfirm, 2\Tconfirm]$ (w.o.p.).

\paragraph{World 3:}
\paragraph{Setup.}
There are two client $\client_1$ and $\client_2$
Validators in $P_2$ are honest and those in $B_s \supseteq P_1 \cup R$ are adversarial.
At round $0$, the environment inputs the transaction $\tx_1$ to the validators in $P_1 \cup R$.
Until round $\Tconfirm$, the validators in $P_1 \cup R$ simulate the execution in world 1 towards $\client_1$.
Simultaneously, validators in $R$ simulate the execution in world 2 towards those in $P_2$ and the client $\client_2$ (split-brain attack), whereas those in $P_1$ ignore all messages from the validators in $P_2$, and do not communicate with them and $\client_2$.
At round $\Tconfirm+1$, the environment inputs the transaction $\tx_2$ to the validators in $P_2 \cup R$.
Validators in $R$ continue to simulate the execution in world 2 towards those in $P_2$ and $\client_2$.
\paragraph{Output.}
Since the worlds 1 and 3 are indistinguishable in $\client_1$'s view (w.o.p.), $\client_1$ outputs the ledger $\langle \tx_1 \rangle$ at round $\Tconfirm$.
Similarly, as the worlds 2 and 3 are indistinguishable in $\client_2$'s view (w.o.p.), $\client_2$ outputs the ledger $\langle \tx_2 \rangle$ at the same round as world 2, \ie within $(\Tconfirm, 2\Tconfirm]$.
However, this implies a safety violation, which is a contradiction.
\end{proof}
\begin{proof}[Proof of Theorem~\ref{thm:psync-safety-liveness-converse}]
Towards contradiction, suppose there exists a protocol with validator set $\UU$, quorum system $\Q$ and fail-prone system $\B_s$ such that $\exists Q^1, Q^2 \in \Q$ and $B_s \in \B_s \colon Q^1 \cap Q^2 \subseteq B_s$.
We will show a safety violation when the validators in $B_s$ are adversarial through the following worlds.
Define $P_1 = Q^1 \backslash Q^2$, $P_2 = (Q^2 \backslash Q^1) \cup (\UU \backslash (Q^1 \cup Q^2))$ and $R = Q^1 \cap Q^2$ (same definition is upheld for the case $Q^1 = Q^2$).

\paragraph{World 1:}
\paragraph{Setup.}
All messages sent by the honest validators are delivered to their recipients in the next round. 
There is a single client $\client_1$.
Validators in $P_1 \cup R$ are honest and those in $P_2$ are adversarial.
The environment inputs a single transaction $\tx_1$ to the validators in $P_1 \cup R$ at round $0$.
The adversarial validators do not communicate with those in $P_1 \cup R$, and do not respond to $\client_1$.
\paragraph{Output.}
By liveness, $\client_1$ outputs the ledger $\langle \tx_1 \rangle$ by round $\Tconfirm$ (w.o.p.).

\paragraph{World 2:}
\paragraph{Setup.}
All messages sent by the honest validators are delivered to their recipients in the next round. 
There is a single client $\client_2$.
Validators in $P_2 \cup R$ are honest and those in $P_1$ are adversarial.
The environment inputs a single transaction $\tx_2$ to the validators in $P_2 \cup R$ at round $0$.
The adversarial validators do not communicate with those in $P_2 \cup R$, and do not respond to $\client_2$.
\paragraph{Output.}
By liveness, $\client_2$ outputs the ledger $\langle \tx_2 \rangle$ by round $\Tconfirm$ (w.o.p.).

\paragraph{World 3:}
\paragraph{Setup.}
World 3 is a hybrid world. 
There are two client $\client_1$ and $\client_2$.
Validators in $P_1$ and $P_2$ are honest and those in $R$ are adversarial.
At round $0$, the environment inputs the transaction $\tx_1$ to the validators in $P_1 \cup R$ and $\tx_2$ to the validators in $P_2 \cup R$.
All messages from the validators in $P_1$ to $P_2$ and $\client_2$ and from those in $P_2$ to $P_1$ and $\client_1$ are delayed by the adversary until after round $\Tconfirm$.
Validators in $R$ do a split-brain attack: One brain simulates the execution in world 1 towards $P_1$ and $\client_1$ with transaction $\tx_1$, and the other simulates the execution in world 2 towards $P_2$ and $\client_2$ with transaction $\tx_2$.
\paragraph{Output.}
Since the worlds 1 and 3 are indistinguishable in $\client_1$'s view (w.o.p.), $\client_1$ outputs the ledger $\langle \tx_1 \rangle$ by round $\Tconfirm$.
Similarly, as the worlds 2 and 3 are indistinguishable in $\client_2$'s view (w.o.p.), $\client_2$ outputs the ledger $\langle \tx_2 \rangle$ by round $\Tconfirm$. 
However, this implies a safety violation when only the validators in $R \subseteq B_s$ are adversarial, which is a contradiction.
\end{proof}
\begin{proof}[Proof of Theorem~\ref{thm:acc-safety-liveness-converse}]
Proof follows the outline of the proof of \cite[Theorem B.1]{forensics}.
Towards contradiction, suppose there exists a protocol with validator set $\UU$, quorum system $\Q$ and fail-prone system $\B_a$ such that $\exists Q^1, Q^2 \in \Q$ and $B_a \in \B_a \colon Q^1 \cap Q^2 \subset B_a$.
Through the following worlds, we will show a world, where an honest validator is identified as a protocol violator by the forensic protocol.
Define $P_1 = Q^1 \backslash Q^2$, $P_2 = (Q^2 \backslash Q^1) \cup (\UU \backslash (Q^1 \cup Q^2))$ and $R = Q^1 \cap Q^2$ (same definition is upheld for the case $Q^1 = Q^2$).
We assume a synchronous network throughout the following worlds.

\paragraph{World 1:}
\paragraph{Setup.}
There is a single client $\client_1$.
Validators in $P_1 \cup R$ are honest and those in $P_2$ are adversarial.
The environment inputs a single transaction $\tx_1$ to the validators in $P_1 \cup R$ at round $0$.
The adversarial validators do not communicate with those in $P_1 \cup R$, and do not respond to $\client_1$.
\paragraph{Output.}
By liveness, $\client_1$ outputs the ledger $\langle \tx_1 \rangle$ by round $\Tconfirm$ (w.o.p.).

\paragraph{World 2:}
\paragraph{Setup.}
There is a single client $\client_2$.
Validators in $P_2 \cup R$ are honest and those in $P_1$ are adversarial.
The environment inputs a single transaction $\tx_2$ to the validators in $P_2 \cup R$ at round $0$.
The adversarial validators do not communicate with those in $P_2 \cup R$, and do not respond to $\client_2$.
\paragraph{Output.}
By liveness, $\client_2$ outputs the ledger $\langle \tx_2 \rangle$ by round $\Tconfirm$ (w.o.p.).

\paragraph{World 3:}
\paragraph{Setup.}
There is a single client $\client_1$.
Validators in $P_1$ are honest and those in $P_2 \cup R$ are adversarial.
At round $0$, the environment inputs the transaction $\tx_1$ to the validators in $P_1 \cup R$ and $\tx_2$ to the validators in $P_2 \cup R$.
Validators in $R$ do a split-brain attack: One brain simulates the execution in world 1 towards $P_1$ and $\client_1$ with transaction $\tx_1$, and the other simulates the execution in world 2 with $P_2$ with transaction $\tx_2$.
Validators in $P_2$ simulate the execution in world 2, do not communicate with those in $P_1$ and $\client_1$, and ignore all messages sent by them.
\paragraph{Output.}
Since the worlds 1 and 3 are indistinguishable in $\client_1$'s view (w.o.p.), $\client_1$ outputs the ledger $\langle \tx_1 \rangle$ by round $\Tconfirm$.
Since the validators in $P_2 \cup R$ simulate the execution in world 2, by emulating the client of world 2, an adversarial validator outputs the ledger $\langle \tx_2 \rangle$ by round $\Tconfirm$.
Since this is a safety violation, a new client $c_3$ asks the validators for their transcripts, upon which the adversarial validators in $P_2$ reply with transcripts that omit the messages received from those in $P_1$.
The client $\client_3$ then runs the forensic protocol with these transcripts and outputs a proof that identifies at least one adversarial validator in $P_2$ as $R \subset B_a$.

\paragraph{World 4:}
\paragraph{Setup.}
There is a single client $\client_2$.
Validators in $P_2$ are honest and those in $P_1 \cup R$ are adversarial.
At round $0$, the environment inputs the transaction $\tx_2$ to the validators in $P_2 \cup R$ and $\tx_1$ to the validators in $P_1 \cup R$.
Validators in $R$ do a split-brain attack: One brain simulates the execution in world 1 with $P_1$ with transaction $\tx_1$, and the other simulates the execution in world 2 towards $P_2$ and $\client_2$ with transaction $\tx_2$.
Validators in $P_1$ simulate the execution in world 1, do not communicate with those in $P_2$ and $\client_2$, and ignore all messages sent by them.
\paragraph{Output.}
Since the worlds 2 and 4 are indistinguishable in $\client_2$'s view (w.o.p.), $\client_2$ outputs the ledger $\langle \tx_2 \rangle$ by round $\Tconfirm$.
Since the validators in $P_1 \cup R$ simulate the execution in world 1, by emulating the client of world 1, an adversarial validator outputs the ledger $\langle \tx_1 \rangle$ by round $\Tconfirm$.
Since this is a safety violation, a new client $c_3$ asks the validators for their transcripts, upon which the adversarial validators in $P_1$ reply with transcripts that omit the messages received from those in $P_2$.
Note that the worlds 3 and 4 (and the transcripts received therein) are indistinguishable in the view of $\client_3$ (w.o.p.).
Thus, upon running the forensic protocol with these transcripts, $\client_3$ outputs a proof that identifies at least one validator in $P_2$.
However, the validators in $P_2$ are honest in world 4, which is a contradiction.
\end{proof}
\begin{proof}[Proof of Theorem~\ref{thm:general-data-limitation}]
Towards contradiction, suppose $\exists Q \in \Q \colon Q \cap \UU' = \emptyset$.
Consider the execution, where all validators in $Q$ are honest and the rest have crashed.
Since the validators in $Q$ do not check for the availability of the proposed blocks, an adversarially proposed valid, yet unavailable block can be output as part of the chain of an honest validator in $Q$.
However, in this case, the external clients checking for the data and receiving this chain lose liveness, implying that the protocol cannot satisfy liveness even when all validators in $Q$ are honest.
This is a contradiction.
\end{proof}

\end{document}